\documentclass[journal]{IEEEtran}
\usepackage{mathrsfs}
\usepackage{latexsym}
\usepackage{graphicx}
\usepackage{epsfig}
\usepackage{subfigure}
\usepackage{array}
\usepackage{amsmath}
\usepackage{amssymb}
\usepackage{color, soul}
\usepackage{amsthm}
\usepackage{enumerate}
\usepackage{algpseudocode}
\usepackage{algorithm}
\usepackage{bm}
\usepackage{amssymb}
\usepackage{balance}
\usepackage{verbatim}
\usepackage{epstopdf}
\usepackage{setspace}
\usepackage{multicol}
\usepackage{amsmath,lipsum}
\usepackage{cuted}
\usepackage{multirow}

\newtheorem{lemma}{Lemma}
\newtheorem{prop}{ Proposition}
\theoremstyle{definition}

\begin{document}

\title{ Robust IRS-aided Secrecy Transmission with Location Optimization
\thanks{Jiale Bai and Hui-Ming Wang are with the School of Information and Communications Engineering, Xi'an Jiaotong University,
Xi'an 710049, China, and also with the Ministry of Education Key Laboratory for Intelligent Networks and Network Security, Xi'an Jiaotong
University, Xi'an 710049, China (e-mail: bjl19970954@stu.xjtu.edu.cn; xjbswhm@gmail.com).
}
\thanks{Peng Liu is with the Wireless Technology Lab, 2012 Labs, Huawei Technologies, Shenzhen, China (e-mail: jeremy.liupeng@huawei.com).
}
\author{Jiale Bai, Hui-Ming Wang, \emph{Senior Member, IEEE} and Peng Liu, \emph{Member, IEEE}
}
}
\maketitle

\begin{abstract}
In this paper, we propose a robust secrecy transmission scheme for intelligent reflecting surface (IRS) aided communication systems. Different from all the existing works where IRS has already been deployed at a fixed location, we take the location of IRS as a variable to maximize the secrecy rate (SR) under the outage probability constraint by jointly optimizing the location of IRS, transmit beamformer and IRS phase shifts with imperfect channel state information (CSI) of Eve, where we consider two cases: a) the location of Eve is known; b) only a suspicious area of Eve is available. We show a critical observation that CSI models are different before and after IRS deployment, thus the optimization problem could be decomposed and solved via a two-stage framework. For case a), in the first stage, universal upper bounds of outage probabilities only related to the location of IRS are derived which can be optimized via successive convex approximation (SCA) method. In the second stage, we develop an alternative optimization (AO) algorithm to optimize beamformer and phase shifts iteratively. For case b), we propose a Max-Min SR scheme based on two-stage framework, where the location of IRS is optimized based on the worst location of Eve. Simulation results indicate the importance of the location of IRS optimization.
\end{abstract}	

\begin{IEEEkeywords}
Physical layer security, intelligent reflecting surface,  secrecy rate, location optimization.
\end{IEEEkeywords}

\section{Introduction}

Intelligent reflecting surface (IRS) has drawn wide attention as an appealing candidate for future wireless communications due to its low hardware cost and energy
consumption \cite{Wu-19}. Specifically, IRS consists of a large number of low-cost passive elements which can induce certain phase shift by a software
controller to incident electromagnetic signal waves, thus making the wireless environment programmable and controllable. By intelligently adapting the phase shifts
of IRS, the reflected signals can be combined constructively at desired receivers to improve the capacity,
spectral and energy efficiency of wireless networks \cite{Wu-19b}\cite{Huang-19}. Furthermore, due to the compact size, IRS can be deployed easily
on buildings or ceilings and integrated into the traditional communication systems with minor modifications. These significant advantages make IRS an
energy-efficient technique that is applied into various communication scenarios such as multi-cell \cite{Xie-21}, wireless information and power transfer design
\cite{Pan-20} and cognitive radio design \cite{Xu-20}. 

Due to the above capability, it is a natural idea to apply IRS to enhance  physical layer security (PLS) of wireless transmissions. By properly adjusting the phase shift, the reflected signal by IRS not only can
be superimposed constructively at Bob, but also superimposed destructively at Eve so that
the propagation channel
link between the transmitter (Alice) and the legitimate user (Bob)/eavesdropper (Eve) can be enhanced/deteriorated,   thus significantly boosting
the secrecy rate (SR). Combining with various schemes, such as multiple-antenna techniques, artificial noise (AN) and cooperative jamming, the IRS can further
enhance secrecy performance.

\subsection{Related works and Motivation}
Motivated by these advantages, IRS-aided secrecy communication has been widely studied in the
literature \cite{Yu-19}-\cite{Wang-19}. Several contributions have been established for improving secrecy performance based on the ideal assumption that perfect channel state information (CSI) of both
Bob and Eve are known at Alice \cite{Yu-19}-\cite{Chu-21}. To maximize the SR, the single-user MISO model was investigated in \cite{Yu-19}-\cite{Zheng-20} by jointly
optimizing the transmit beamformer and phase shifts. Later, the model was extended to a MISO multi-user scenario \cite{Chen-19}\cite{Xu-19}. Furthermore, authors
in \cite{Dong-20b}-\cite{Chu-21} consider an MIMO single-user secrecy transmission problem. Due to the introducing of multiple antennas at both transmitter and receiver, the joint optimization problem becomes more complicated and sophisticated algorithms have been proposed to solve it.

In practice, Alice is impossible to acquire the perfect CSI of Eves so that several contributions investigated the secrecy transmission based on the assumption
with imperfect CSI of Eves \cite{Hong-21}-\cite{Yu-20}, even no CSI of Eves \cite{Dong-20b}\cite{Wang-19}. The authors in \cite{Hong-21}\cite{Xu-21}
considered a statistical CSI error of cascaded channels of Alice-IRS-Eve. In \cite{Hong-21}, the beamformer, AN covariance matrix, and IRS phase shifts are jointly optimized to minimize the transmit power subject to outage probability constraint. In \cite{Xu-21}, IRS was deployed to modulate the received signals into jamming signals so as to meet the quality-of-service (QoS) requirement of Bob and deteriorate the reception at Eves. In contrast to a statistical CSI error model \cite{Hong-21}\cite{Xu-21}, a bounded CSI error is considered in \cite{Lu-20}\cite{Yu-20}. The multiple-users MISO \cite{Lu-20} and multiple-users MISOME models \cite{Yu-20} were investigated respectively. What's more, authors in \cite{Fang-21} considered statistical CSI of Eve to maximize the ergodic secrecy rate by jointly designing the trajectory, power control of UAV and phase shifts of IRS in IRS-aided UAV SISO system. Furthermore, in more practical applications without CSI of Eve completely, \cite{Dong-20b} and \cite{Wang-19} developed an AN joint transmission scheme in MIMO and MISO systems, respectively. The results in all aforementioned works \cite{Yu-19}-\cite{Wang-19} unveiled that IRSs can significantly improve secrecy performance compared with conventional architectures without IRS.

However, all the aforementioned works \cite{Yu-19}-\cite{Wang-19} are restricted to the case that the \emph{ location of IRS is fixed}. In fact, where the IRSs are deployed is also very important to the secrecy performance. The works in \cite{Hu-20}-\cite{Mu-21} revealed the importance of IRS location optimization for improving communications rate. But there is still
no work considering IRS location optimization in secrecy transmissions. Since that the location of IRS affects both amplitudes and phases of the cascaded
channels which will impact the achievable secrecy rate, the difficulty of tackling the joint optimization problem by taking the location of IRS into consideration  increases significantly compared with the conventional case with
a fixed IRS location.
Furthermore, since eavesdropping is a passive attack so it is more practical to assume that only some statistical CSI of Eve, such as channel distribution information (CDI), is available rather than complete CSI. In the worst case, even the location of Eve is unknown at all but only some suspicious area is known where an Eve may exist. This is a typical scenario of many practical applications.

Therefore, how to provide a PLS transmission in these more practical scenarios is very important for the real applications of IRS into PLS, which motivates this work.

\subsection{ Approaches and Contributions}
In this paper, we design a robust secrecy transmission to maximize the SR under the constraint of secrecy outage probability by jointly optimizing location of IRS, transmit beamformer at Alice and phase shifts at IRS with imperfect CSI of Eve. In particular, two cases are considered: a) The location of Eve is known so the line-of-sight (LoS) component in the cascaded Alice-IRS-Eve link is available and only the CDI of the non-line-of-sight (NLoS) component is known; b) The location of Eve is unknown so neither LoS nor NLoS component is known and only a suspicious area of Eve and its CDI is available. To the best of our knowledge, this is the first work to optimize the location of IRS for robust secrecy transmission and this is also the first work to maximize the SR with CDI of Eve.
The main contributions are summarized as follows.

1) We establish the joint optimization problem when  location of Eve is known, and show a critical observation that the prior knowledge of CSI depends on the location knowledge of IRS, which has a significant impact on the underlying optimization problem. In other words, the objective function appears differently before and after IRS deployment.
This important feature invalidates conventional alternating optimization (AO) algorithm to solve the joint optimization problem.
To tackle this difficulty, a two-stage algorithm is proposed, where it optimizes the location of IRS in the first stage
and jointly optimizes beamformer and phase shifts in the second stage. This framework is the most distinctive feature compared with existing works with fixed IRS location.

2) For the first stage of optimizing the IRS location, the difficulty lies in how to handle the secrecy outage probability constraint without the beamformer at Alice and phase shifts at IRS. We tackle this problem by deriving universal upper bounds of outage probabilities for both Bob and Eve which are only related to the deployment location of IRS. In such a way, IRS location could be optimized by solving a non-convex problem via successive convex approximation (SCA) method. Then in the second stage, to handle the outage probability constraint with only CDI of Eve, we exploit Bernstein-Type Inequality (BTI) to reshape the non-convex constraints, and take an AO procedure to optimize beamformer at Alice and phase shifts at IRS iteratively.

3) In the case where the exact location of Eve is unknown but only some suspicious area is known where an Eve may exist,  we propose a Max-Min SR scheme based on the two-stage method, where in the first stage we optimize the location of IRS based on the worst location of Eve resulting in the minimized SR. The simulation results indicate the convergence and effectiveness of proposed algorithms and show that the location optimization of IRS plays a crucial role in the secrecy performance. In addition, simulation results also show that wherever Eve is located, IRS should be deployed nearby Bob to reduce large-scale path loss.

The rest of the paper is organized as follows: Section II describes the signal transmission, CSI model and problem formulate. In section III,  the algorithm is proposed to maximize SR with the location of Eve. The case without location of Eve is investigated in section IV. Simulation results have been carried out to evaluate the performance and
convergence of proposed algorithm in section V. Finally, section VI concludes the paper.

\emph{Notations}: ${\bf A}^{\rm T}$ and ${\bf A}^{\rm H}$ denote transpose and Hermitian conjugate of ${\bf A}$, respectively; $\mathbb{E}\left \{ \cdot  \right \}$
 is statistical expectation; ${\rm \rm Tr}({\bf A})$, ${\rm{Re}}\{{\bf A}\}$ and ${\rm rank}({\bf A})$ denote the trace, real elements and rank of ${\bf A}$; ${ \rm vec}({\bf A})$
 is the vector obtained by stacking all columns of matrix ${\bf A}$ on top of each other; $\|{\bf a}\|$ denotes the norm of the vector ${\bf a}$; ${\rm arg}({\bf a})$ denotes
 the phase shifts of ${\bf a}$; ${\rm diag}({\bf a})$ is to transform the vector ${\bf a}$ as a diagonal matrix with diagonal elements in ${\bf a}$; ${\rm Diag}({\bf A})$ denotes a vector whose elements are extracted from the main diagonal elements of matrix ${\bf A}$; ${\bf 1}_M$ present a column vector with $M$ ones; $\otimes $ represents the Kronecker product.

\section{System Model and Problem Formulation}

In this section, we present the signal transmission model, channel model and problem formulation  in a secure IRS-aided communication system.

\subsection{ Signal Transmission Model }\label{system}

\begin{figure}[t]
	\centerline{\includegraphics[width=3.4in]{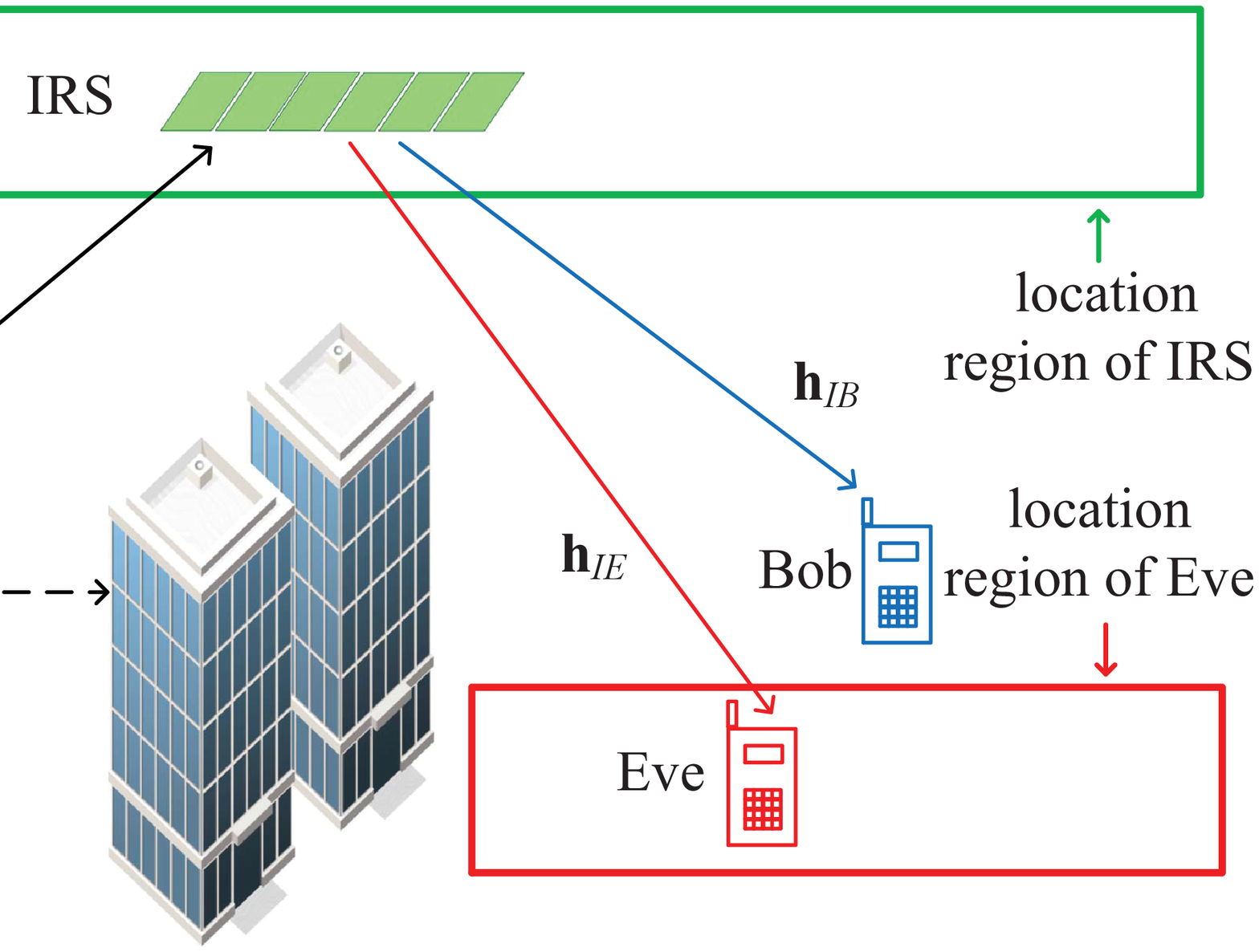}}
	\caption{ A planform of an IRS-assisted MISO secure system. The location area of Bob and Eve is denoted in the system.}
\end{figure}

Consider an IRS-assisted communication system as shown in Fig.1, which is consisted of a transmitter (Alice), a legitimate user (Bob),
an eavesdropper (Eve) and an IRS. We assume that both Alice and IRS use uniform linear arrays (ULA)
with $N_t$ antennas and $M$ reflection elements, respectively, while Bob and Eve are single antenna nodes. The direct link between Alice and Bob/Eve is
blocked due to obstacles, which is a typical scene where IRS is deployed to overcome blockage.  Without loss of generality,
we assume that Alice is located at the origin of a Cartesian coordinate system and the ULAs of Alice and IRS are along x-axis, hence the locations of Alice, IRS, Bob, and Eve are
$\boldsymbol{\omega}_{A}=\left (0,0\right)^T$, $\boldsymbol{\omega}_{I}=\left (x_I,y_I\right)^T$,  $\boldsymbol{\omega}_{B}=
\left (x_B,y_B\right)^T$ and $\boldsymbol{\omega}_{E}=\left (x_E,y_E\right)^T$, respectively. Let $\Omega_I$ specify a predefined area for deploying IRS\footnote{In general, we can find a larger rectangular area to contain IRS location area whatever geometry it is. For example, if IRS location area is circle, we search the smallest rectangle to contain the circle. After optimizing $\boldsymbol{\omega}_I$, if it is outside the circle, we search for the closest location to $\boldsymbol{\omega}_I^*$ as the optimal location of IRS. }
\begin{align}
	\notag
	\boldsymbol{\omega}_I \in \Omega_I \triangleq\left\{ \left(x_I,y_I\right)^T | x_I\in \Delta_{x_I}, y_I \in \Delta_{y_I}\right\},
\end{align}
where $\Delta_{x_I}\triangleq\left[x^{min}_I, x^{max}_I\right]$ and $\Delta_{y_I}\triangleq\left[y^{min}_I, y^{max}_I\right]$ denote the candidate deploying areas along x- and y-axes. When the location of Eve is unknown, in this paper we assume that Eve exists in a suspicious area denoted as
$\Omega_{E}\triangleq\left\{ \left(x_E,y_E\right)^T | x_E\in {\Delta}_{x_E}, y_E \in {\Delta}_{y_E}\right\}$, i.e., $\boldsymbol{\omega}_E \in \Omega_E$, where $\Delta_{x_E}\triangleq\left[x^{min}_E, x^{max}_E\right]$ and
$\Delta_{y_E}\triangleq\left[y^{min}_E, y^{max}_E\right]$ are suspicious ranges along the x- and y-axes. This is a typical scenario of many applications. For example, in home wireless access such as WiFi in an apartment, suspicious area is the neighbors rooms which is a priori known.

We consider quasi-static flat-fading channels between all nodes. The received signals at Bob
and Eve are expressed as
\begin{align}
y_b = {\bf{h}}_{IB}^H{\boldsymbol{\Phi}}{\bf{H}}_{AI}{\bf{f}}x + n_b,\
y_e = {\bf{h}}_{IE}^H{\boldsymbol{\Phi}}{\bf{H}}_{AI}{\bf{f}}x + n_e,
\end{align}
where $x$ is the transmitted signal following $\mathbb{E}\left\{|x|^2 \right\}=1$, ${\bf{f}}\in\mathbb{C}^{N_t\times 1}$ represents the beamformer of Alice, ${\bf{H}}_{AI}\in\mathbb{C}^{M\times N_t}, {\bf{h}}_{IB}\in\mathbb{C}^{M\times 1}$ and ${\bf{h}}_{IE}\in\mathbb{C}^{M\times 1}$ are the flat fading channel links of Alice-IRS, IRS-Bob and IRS-Eve respectively, $n_B$, $n_E$ represent noise at Bob and Eve, respectively, in which the entries are with zero-mean and variances ${\sigma_b}^2$
and ${\sigma_e}^2$ and we set ${\sigma_b}^2={\sigma_e}^2={\sigma}^2$. ${\boldsymbol{\Phi}} \triangleq {\rm diag}(e^{j\theta_1},\cdots,e^{j\theta_M})$ is the diagonal phase shift matrix for IRS with $e^{j\theta_i}$ the phase shift coefficient at reflecting element $i$ ($i=1,\cdots,M$). In addition, a controller is used to coordinate Alice and
IRS for channel acquisition and data transmissions \cite{Wang-20}.

The equivalent channels of Bob and Eve could be re-written as:
\begin{align}
{\bf{h}}_{IJ}^H\boldsymbol{\Phi}{{\bf{H}}}_{AI}{\bf{f}}=\boldsymbol{\phi}^H{{\bf{G}}}_{AJ}{\bf{f}},
\end{align}
where ${{\bf{G}}}_{AJ} \triangleq{\rm diag}\left({\bf{h}}_{IJ}\right){{\bf{H}}}_{AI}\in \mathbb{C}^{M\times N_t}, J\in \left\{B,E\right\}$ is the cascaded channel
from Alice to Bob or Eve via IRS, and $\boldsymbol{\phi}\triangleq(e^{j\theta_1},\cdots,e^{j\theta_M})^T \in
\mathbb{C}^{M\times 1}$.

\subsection{Channel Model}\label{channel_model}
Since IRS is usually deployed on high places with line-of-sight (LoS) path, we use angle-domain Rician fading to model all the channels.
Specifically, the channel from Alice to IRS can be expressed as
 \begin{align}
\notag
{\bf{H}}_{AI} \triangleq \sqrt {L_{AI}}{\left( \sqrt{\frac{\kappa}{\kappa+1}}{{{\bf{H}}}}_{AI}^{LoS} + \sqrt{\frac{1}{\kappa+1}}
{\bf{H}}_{AI}^{NLoS}\right)},
\end{align}
where $L_{AI}\triangleq L_0d_{AI}^{-\rho_{AI}}$, $\frac{1}{L_{AI}}$ is the large-scale path loss of Alice-IRS, $L_0 \triangleq\left( \frac{\lambda_c}{4\pi}\right)^2$ is constant with
$\lambda_c$ being wavelength of the center frequency of carrier \cite{Hong-21}\cite{Wu-19c}, $d_{AI}\triangleq\|\boldsymbol{\omega}_{A}-\boldsymbol{\omega}_{I}\|$
is the distance from Alice to IRS, while $\rho_{AI}$ is the corresponding path loss exponent respectively. The small-scale fading in Alice-IRS link
is assumed to be Rician fading with Rician factor $\kappa$, where LoS and NLoS components are denoted as ${{\bf{H}}}_{AI}^{LoS}$ and
${{\bf{H}}}_{AI}^{NLoS}$, respectively.  ${{\bf{H}}}_{AI}^{LoS}$ is denoted as
\begin{align}
\notag
{{\bf{H}}}_{AI}^{LoS} \triangleq& \boldsymbol{\alpha}_I(\theta_{AI})\boldsymbol{\alpha}_A^H(\varphi_{AI}),
\end{align}
where $\boldsymbol{\alpha}_A(\varphi_{AI})\in \mathbb{C}^{N_t\times 1}$ and $\boldsymbol{\alpha}_I(\theta_{AI})\in \mathbb{C}^{M\times 1}$ are array response vectors
at Alice and IRS with effective angles of departure (AOD) $\varphi_{AI} = {\rm{arccos}} \left(\frac{x_I}{||\boldsymbol{\omega}_I||}\right)$ and arrival (AOA) $\theta_{AI}= \pi - \varphi_{AI}$,
\begin{align}
\notag
\boldsymbol{\alpha}_A(\varphi_{AI})& \triangleq \left[1, \ \cdots ,\ e^{j(N_t-1)2\pi \frac{d_a}{\lambda_c} {\rm{cos}}
(\varphi_{AI})} \right]^H,\\
\notag
 \boldsymbol{\alpha}_I(\theta_{AI})& \triangleq \left[1,\ \cdots ,\ e^{j(M-1)2\pi \frac{d_r}{\lambda_c} {\rm{cos}}
(\theta_{AI})} \right]^H,
\end{align}
respectively where $d_a$ and $d_r$ are the distance between two adjacent Alice antennas and IRS elements and we set $d_a=d_r=\frac{\lambda_c}
{2}$. ${{\bf{H}}}_{AI}^{NLoS}$ is Rayleigh fading with entries followed $\mathcal{CN}(0,1)$.

Similarly, the channel from IRS to Bob ${\bf{h}}_{IB}$ and to Eve ${\bf{h}}_{IE}$ are given by
\begin{align}
\notag
{\bf{h}}_{IJ} &\triangleq \sqrt {L_{IJ}}{\left(\sqrt{\frac{\kappa}{\kappa+1}}{\bf{h}}_{IJ}^{LoS} + \sqrt{\frac{1}{\kappa+1}}{\bf{h}}_{IJ}^
{NLoS}\right)},\ J\in \left\{B,E\right\},
\end{align}
where $L_{IJ} \triangleq L_0d_{IJ}^{-\rho_{IJ}}$, $\frac{1}{L_{IJ}}$ is the large-scale path loss of IRS-Bob $(J=B)$ or IRS-Eve $(J=E)$ link,
$d_{IJ}\triangleq\|\boldsymbol{\omega}_{I}-\boldsymbol{\omega}_{J}\|$ is the distances of IRS-Bob or IRS-Eve with path loss exponents $\rho_{IJ}$. We set $\rho_{IB}=\rho_{IE}=\rho$. The small-scale fading is also assumed to be Rician fading with Rician factor $\kappa$. The LoS component
 ${\bf{h}}_{IJ}^{LoS}$ is given by:
\begin{align}
\notag
{\bf{h}}_{IJ}^{LoS} &\triangleq\boldsymbol{\alpha}_I(\varphi_{IJ})\triangleq\left[1, \ \cdots, \ e^{j(M-1)\pi {\rm{cos}}(\varphi_{IJ})} \right]^H,
 \ J\in \left\{B,E \right\},
\end{align}
where $\boldsymbol{\alpha}_I(\varphi_{IJ})\in \mathbb{C}^{M\times 1}$ is array response vector of IRS-Bob or IRS-Eve links, AoDs of which is denoted as $\varphi_{IJ}$,
i.e., $\varphi_{IJ}= {\rm{arccos}}\left(\frac{x_J-x_I}{\|\boldsymbol{\omega}_J-\boldsymbol{\omega}_I\|}\right)$. All the entries of NLoS components in ${\bf{h}}_{IJ}^{NLoS}$ are distributed as $\mathcal{CN}(0,1)$.

 We have to emphasize that different from the most existing works, in this work the location of IRS is not fixed but a parameter we want to optimize. Therefore, in our problem, the beamformer $\bf f$ and coefficients $\boldsymbol{\phi}$ are optimization variables adapted to the small-scale fading, and the location  $\boldsymbol{\omega}_{I}$  is also a variable adapted to the large-scale fading.
 Based on the above models, we have an critical observation that there are two CSI models with and without IRS location, i.e., the CSI models are different before and after IRS deployment.

1) When the location of IRS $\boldsymbol{\omega}_{I}$ is an optimizing variable that have not been determined, the path loss components $L_{AI}$, $L_{IJ}$, and LoS components ${\bf{H}}_{AI}^{LoS}$, ${\bf{h}}_{IJ}^{LoS}, J\in \left\{B,E\right\}$  are all functions of this variable, but all the instantaneous NLoS components  ${\bf{H}}_{AI}^{NLoS}$, ${\bf{h}}_{IJ}^{NLoS},  J\in \left\{B,E\right\}$ should be taken as unknown random variables whose entries follow
$\mathcal{CN}(0,1)$. This is due to the fact that in order to estimate instantaneous CSI, the location of IRS must be determinated first and  channel training could then be performed. In this case, we retransform the CSI model of all the channels as
\begin{align}
\notag
{\bf{H}}_{AI}& =\underbrace{\sqrt{\frac{\kappa L_{AI}}{\kappa+1}}{\bf{H}}_{AI}^{LoS}}_{\overline{ {\bf{H}}}_{AI}}+\underbrace{\sqrt{\frac{L_{AI}}{\kappa+1}}{\bf{H}}_{AI}^{NLoS}}_{ \widetilde{\bf{H}}_{AI}},\\
\notag
{\bf{h}}_{IJ}& =\underbrace{\sqrt{\frac{\kappa L_{IJ}}{\kappa+1}}{\bf{h}}_{IJ}^{LoS}}_{\overline{ {\bf{h}}}_{IJ}}+ \underbrace{\sqrt{\frac{L_{IJ}}{\kappa+1}}{\bf{h}}_{IJ}^{NLoS}}_{ \widetilde{\bf{h}}_{IJ}},\ J\in \left\{B,E \right\},
\end{align}
where $\overline{ {\bf{H}}}_{AI}$ and $\overline{ {\bf{h}}}_{IJ},J\in \left\{B,E \right\}$ denote the known components in Alice-IRS and IRS-Bob/Eve link; $\widetilde{\bf{H}}_{AI}$ and $\widetilde{\bf{h}}_{IJ}$ denote the random components whose entries follow $\mathcal{CN}\left(0,\frac{L_{AI}}{\kappa+1}\right)$ and $\mathcal{CN}\left(0,\frac{L_{IJ}}{\kappa+1}\right)$. Then, the cascaded channel ${\bf{G}}_{AB}$ and ${\bf{G}}_{AE}$ is re-expressed as
\begin{align}
 \notag
{\bf G}_{AJ}&={\left[{{\rm diag}\left(\overline{ \bf{h}}_{IJ} \right) + {\rm diag}\left(\widetilde{\bf{h}}_{IJ}\right)}\right]}\left(\overline{ \bf{H}}_{AI}+\widetilde{\bf{H}}_{AI}\right)\\
\notag
&\overset{(b)}{\approx } \underbrace{{\rm diag}\left(\widetilde{\bf{h}}_{IJ}\right)\overline{\bf{H}}_{AI} + {\rm diag}\left(\overline{ \bf{h}}_{IJ} \right)\widetilde{\bf{H}}_{AI}}_{\widetilde{\bf{G}}_{AJ}}\\
\label{AllError}
&+\underbrace{{\rm diag}\left(\overline{ \bf{h}}_{IJ} \right)\overline{\bf{H}}_{AI}}_{\overline{ \bf{G}}_{AJ}}, \ J\in \left\{B,E \right\},
\end{align}
 where $(b)$ is obtained due to the fact that IRS is usually deployed in high places to avoid signal blockage, therefore the LoS component in Alice-IRS channel is expected to be the dominant factor and the product of the two random components ${\rm diag}\left(\widetilde{\bf{h}}_{IJ}\right)\widetilde{\bf{H}}_{AI}$ is too small, which could be ignored especially when the Rician factor is large. Now in $\eqref{AllError}$, $\overline{ \bf{G}}_{AJ}$ are determinate and known, while $\widetilde{\bf{G}}_{AJ}$ are random. Then, we analyze $\overline{ \bf{G}}_{AJ}$ and $\overline{\bf{H}}_{AI}$ in detail. Due to each element of $\overline{ \bf{h}}_{IJ}$ and $\overline{\bf{H}}_{AI}$ with constant modulus $\sqrt{\frac{\kappa L_{IJ}}{\kappa+1}} $ and  $\sqrt{\frac{\kappa L_{AI}}{\kappa+1}}$, so each element of $\overline{ \bf{G}}_{AJ}$ also has the constant modulus $\frac{\kappa \sqrt{ L_{AI}L_{IJ}}}{\kappa+1}$. Proposition 1 implies that each element of $\widetilde{\bf{G}}_{AJ}$ follows $\mathcal{CN}\left(0,\frac{2\kappa L_{AI}L_{IJ}}{\left(\kappa+1\right)^2}\right)$. Hence, the cascaded channel $\widetilde{\bf{G}}_{AJ}$ still is a Rician channel.
 \vspace{0ex}
 \begin{prop}
If $\widetilde{\bf{H}}_{AI}$ and $\widetilde{\bf{h}}_{IJ}, J\in \left\{B,E \right\}$ follow cyclic symmetric complex Gaussian (CSCG) distribution and $\overline{\bf{H}}_{AI}$ and $\overline{\bf{h}}_{IJ}$ are determinate,
$\widetilde{\bf{G}}_{AJ}\triangleq{\rm diag}\left(\widetilde{\bf{h}}_{IJ}\right)\overline{\bf{H}}_{AI} + {\rm diag}\left(\overline{ \bf{h}}_{IJ} \right)\widetilde{\bf{H}}_{AI}$ is still CSCG distribution whose elements follow $\mathcal{CN}\left(0,\frac{2\kappa L_{AI}L_{IJ}}{\left(\kappa+1\right)^2}\right)$.
\end{prop}
\vspace{0ex}
\begin{proof}
See Appendix A.
\end{proof}

2) When the location of IRS $\boldsymbol{\omega}_{I}=\left (x_I,y_I\right)^T$ is fixed, the instantaneous CSI of ${\bf{H}}_{AI}$ and ${\bf{h}}_{IB}$ are available to Alice via pilot channel training and estimation, and $L_{IE}$ and ${\bf{h}}_{IE}^{LoS}$ in IRS-Eve links are also available since the location of Eve is known. Now only ${\bf{h}}_{IE}^{NLoS}$ are random variables following Rayleigh distribution. The IRS-Eve link is reexpressed as ${\bf{h}}_{IE} = \overline{ {\bf{h}}}_{IE} + \widetilde{\bf{h}}_{IE}$,
where $\overline{ {\bf{h}}}_{IE} \triangleq\sqrt{\frac{\kappa L_{IE}}{\kappa+1}}{\bf{h}}_{IE}^{LoS}$ is determinate component and $\widetilde{\bf{h}}_{IE}\triangleq\sqrt{\frac{L_{IE}}{\kappa+1}}{\bf{h}}_{IE}^{NLoS}$ is random component. Then, the cascaded channel ${{\bf{G}}}_{AE}$ is reexpressed as
\begin{align}
\label{EveError}
{\bf G}_{AE}&={\left[{{\rm diag}\left(\overline{ \bf{h}}_{IE} \right) + {\rm diag}\left(\widetilde{\bf{h}}_{IE}\right)}\right]}
{\bf{H}}_{AI}\triangleq\overline{ {{\bf{G}}}}_{AE} + \widetilde{{\bf{G}}}_{AE},
\end{align}
where $\overline{ \bf{G}}_{AE}\triangleq{\rm diag}\left(\overline{ \bf{h}}_{IE} \right){\bf{H}}_{AI}$ are determinate, while $\widetilde{\bf{G}}_{AE}\triangleq{\rm diag}\left(\widetilde{\bf{h}}_{IE}\right){\bf{H}}_{AI}$ are random and only its CDI is known. Similarly as (\ref{AllError}), according to proposition 1, we derive that i.d.d elements of $\widetilde{{\bf{G}}}_{AE}$ distribute as $\mathcal{CN}\left(0,\frac{\kappa L_{AI}L_{IE}}{(\kappa+1)^2}\right)$.

The above two observations on the CSI models are critical in establishing and solving the following optimization problem, as will be seen later.

When the location of Eve is unknown but only some suspicious area is known where an Eve may exist, neither LoS nor NLoS component in IRS-Eve links is known. Since $\widetilde{{\bf{G}}}_{AE}$ in both (\ref{AllError}) and (\ref{EveError}) is related to location of IRS and Eve, for this case, the basic idea is to
find the deployment location and the worst location of Eve together based on the CSI model as (\ref{AllError}) at first. Once the locations of IRS and Eve are determined, the following procedure will have the CSI model as (\ref{EveError}).

\subsection{ Problem Formulation}
In this paper, we aim to jointly optimize the location of IRS, transmit beamformer at Alice, phase shifts at IRS to maximize the achievable secrecy rate (SR) under the outage probability constraint in two cases where Alice knows the location of Eve and suspicious area of Eve. Based on the above model, the achievable rates of Bob and Eve are expressed as
\begin{align}
\label{Rb}
C_J(\boldsymbol{\omega}_I,{\bf{f}},\boldsymbol{\phi})&= \log\left(1+\frac{|\boldsymbol{\phi}^H{{\bf{G}}}_{AJ}{\bf{f}}|^2}{{\sigma}^2}\right),J\in\left\{B,E\right\}
\end{align}
the achievable SR $R_s$ is $R_s(\boldsymbol{\omega}_I,{\bf{f}},\boldsymbol{\phi})\triangleq C_B(\boldsymbol{\omega}_I,{\bf{f}},\boldsymbol{\phi})-C_E(\boldsymbol{\omega}_I,{\bf{f}},\boldsymbol{\phi})$. The NLoS component in ${\bf{G}}_{AE}$ is unknown so we use outage probability to evaluate the secrecy performance.

1) When the location of Eve is perfectly known at Alice, the problem can be formulated as
\begin{subequations}\label{ PerfectLocation}
\begin{align}
\max \limits_{{\bf{f}},\boldsymbol{\phi},\boldsymbol{\omega}_I} &\ R \label{ PerfectLocationa}\\
 {\rm{s.t.}} &\ {\rm{Pr}}\left\{R_s(\boldsymbol{\omega}_I,{\bf{f}},\boldsymbol{\phi})\leq R\right\} \leq p_{out},\label{ PerfectLocationb}\\
&\ \|{\bf{f}}\|^2 \leq P,\ |\boldsymbol{\phi}_{i}|=1,\ i=1,\cdots,M,\ \boldsymbol{\omega}_I \in \Omega_{I}, \label{ PerfectLocationc}
\end{align}
\end{subequations}
where  $R$ are the target SR, $\left(\ref{ PerfectLocationb}\right)$ can be seen as the secrecy outage constraint with outage probability $p_{out}$, $\|{\bf{f}}\|^2 \leq P$ is the total power constraint at Alice with total transmit power budget $P$, the unit modulus constraint $|\boldsymbol{\phi}_{i}|=1$ ensures that each reflecting element in IRS does not change the amplitude of signal and $\boldsymbol{\omega}_I \in \Omega_{I}$ represents the location area of IRS.

2) When the location of Eve is unknown, we denote a suspicious area of Eve which is a suspicious area where Eve may exist. To guarantee the security, we consider the worst case that  the information leakage to Eve is maximum, which means we want to maximize the minimum SR. The problem can be formulated as
\begin{align}
\label{NL}
\max \limits_{{\bf{f}},\boldsymbol{\phi},\boldsymbol{\omega}_I} \ \min \limits_{\boldsymbol{\omega}_E} &\ R,
{\rm{s.t.}} \ (\ref{ PerfectLocationb}),\ (\ref{ PerfectLocationc}),\ \boldsymbol{\omega}_E \in \Omega_{E},
\end{align}
where $\boldsymbol{\omega}_E \in \Omega_{E}$ represents the constraint of suspicious area of Eve.

\section{ Problem Solution With Perfect Location Of Eve}\label{Two-Stage}
In this section, we propose an efficient algorithms to maximize SR under the case where Alice knows the location of Eve. For problem (\ref{ PerfectLocation}),
we find that both the objective function (\ref{ PerfectLocationa}) and each constraint in (\ref{ PerfectLocationc}) are only related with one optimization variable. However, for the constraint (\ref{ PerfectLocationb}), $\boldsymbol{\omega}_I$ impacts not only  the large-scale path loss $L_{AI}, L_{IB}, L_{IE}$ but also LoS components ${\bf{H}}_{AI}^{LoS}, {\bf{h}}_{IB}^{LoS}, {\bf{h}}_{IE}^{LoS}$. All the optimization variables are coupled in (\ref{ PerfectLocationb}). At a first glance, it seems that an alternating optimization (AO) algorithm could be tried to solve the optimization problem, where in each iteration only one optimization variable is optimized by taking the other variables fixed as the values at the last iteration. However, we will show that this is not the fact.

As described in Section \ref{channel_model}, for optimizing $\boldsymbol{\omega}_I$ when it is taken as an optimization variable, the instantaneous CSI ${\bf{H}}_{AI}^{NLoS}, {\bf{h}}_{IB}^{NLoS}, {\bf{h}}_{IE}^{NLoS}$ are unavailable for Alice, which should be taken as random variables. In this case the CSI model is expressed as (\ref{AllError}). However, once $\boldsymbol{\omega}_I$ is fixed and for optimizing ${\bf f}$ and $\boldsymbol{\phi}$, only ${\bf{h}}_{IE}^{NLoS}$ is unavailable and the CSI model now is expressed as (\ref{EveError}). In a word, for optimizing different variables, the secrecy outage constraint (\ref{ PerfectLocationb}) has different CSI models. Hence, the AO method could not be applied directly to solve this optimization problem. Therefore, we show a critical observation that CSI models are different before and after IRS deployment, thus the joint optimization problem  (\ref{ PerfectLocation}) could be naturally solved thought a two-stage optimization framework: first to search for an appropriate location of IRS, and  then fix IRS at this location, estimates the instantaneous CSI of Alice-IRS-Bob link, and finally optimize the beamformer of Alice and phase shifts of IRS.

Mathematically, in the first stage, we optimize $\boldsymbol{\omega}_I$ to obtain an optimal location of IRS with the CSI model shown in (\ref{AllError}). Then, in the second stage, we adjust ${\bf f}$ and $\boldsymbol{\phi}$ to maximize the SR based on this optimized $\boldsymbol{\omega}_I$ subject to all the constraints with the CSI model shown in (\ref{EveError}).

\subsection{ Optimize Location of IRS}\label{Loca}

In this subsection, we optimize $\boldsymbol{\omega}_I$ to obtain an optimal location of IRS. When $\boldsymbol{\omega}_I$ is a variable, only distribution of NLoS components in all the channels could be available and the CSI model is described in (\ref{EveError}). The sub-problem for optimizing $\boldsymbol{\omega}_I$ is expressed as
\begin{align}
\max \limits_{\boldsymbol{\omega}_I \in \Omega_{I}} &\ R ,\
 {\rm{s.t.}} \ {\rm{Pr}}\left\{R_s(\boldsymbol{\omega}_I,{\bf{f}},\boldsymbol{\phi})\leq R \right\} \leq p_{out},\label{location1}
\end{align}
where $R_s(\boldsymbol{\omega}_I,{\bf{f}},\boldsymbol{\phi})\triangleq C_B(\boldsymbol{\omega}_I,{\bf{f}},\boldsymbol{\phi})-C_E(\boldsymbol{\omega}_I,{\bf{f}},\boldsymbol{\phi})$. Recalling (\ref{Rb}), since both ${\bf{G}}_{AB}$ and ${\bf{G}}_{AE}$ have random components, both $C_B(\boldsymbol{\omega}_I,{\bf{f}},\boldsymbol{\phi})$ and $C_E(\boldsymbol{\omega}_I,{\bf{f}},\boldsymbol{\phi})$ in $R_s(\boldsymbol{\omega}_I,{\bf{f}},\boldsymbol{\phi})$ have random components. According to Theorem 1 in \cite{Yan-16}, the secrecy outage probability for given a $R$ is shown in (\ref{SOP}), where $\Gamma(\cdot)$ is Gamma function, $\Gamma_G(\cdot)$ is the generalized gamma function, and all the parameters $\widetilde{m}_{B}$, $\widetilde{m}_{E}$, $\widetilde{\gamma}_{B}$, $\widetilde{\gamma}_{E}$, $\widetilde{\bar{\gamma}}_{B}$ and $\widetilde{\bar{\gamma}}_{E}$ are
 functions of  locations of Bob and Eve, beamformer and phase shift. The details are shown in \cite{Yan-16}. We find that the secrecy outage probability involves two infinite series, so problem (\ref{location1}) is hard to solve directly, if not impossible. To make this problem tractable, we transform it as
 \begin{subequations}\label{AO_method1}
 \begin{align}
\max \limits_{\boldsymbol{\omega}_I \in \Omega_{I}} &\ R_B-R_E,\\
\notag
 \ {\rm{s.t.}} &\ {\rm{Pr}}\left\{C_E(\boldsymbol{\omega}_I,{\bf{f}},\boldsymbol{\phi})\geq R_E \right\} \leq p_{out},\\ &{\rm{Pr}}\left\{C_B(\boldsymbol{\omega}_I,{\bf{f}},\boldsymbol{\phi}) \leq R_B\right\} \leq p_{out},\label{out}
\end{align}
\end{subequations}
where $R_B$ and $R_E$ are two newly auxiliary optimizing variables, which represent the target rates of Bob and Eve. The constraint (\ref{out}) is more strict than (\ref{location1}). This is because the constraint (\ref{location1}) includes the data transmission error of Bob where $C_B(\boldsymbol{\omega}_I,{\bf{f}},\boldsymbol{\phi})\leq R_B$, hence the outage probability of it is larger than that of (\ref{out}). It means that for given any $R_B$, $R_E$ and $p_{out}$, we always obtain the optimal $R\leq R_B-R_E$. If problem (\ref{out}) holds, problem (\ref{location1}) must hold. Hence we use (\ref{out}) to replace constraint in (\ref{location1}) which provides an achievable SR lower bound of the original problem.

\begin{figure*}[ht]
\begin{align}
\label{SOP}
{\rm{Pr}}\left\{R_s\leq R \right\}&=\frac{\tilde{m}_{B}^{\tilde{m}_{B}} \tilde{m}_{E}^{\tilde{m}_{E}} 2^{\widetilde{m}_{B} R}}{\Gamma\left( \tilde{m}_{E}\right) \widetilde{\gamma}_{B}^{- \widetilde{m}_{E} \tilde{\gamma}_{E}^{-\widetilde{m}_{B}}} \sum_{n=0}^{+\infty} \frac{2^{n R} \exp \left(-\frac{\widetilde{m}_{B}\left(2^{R}-1\right)}{\tilde{\gamma}_{B}}\right)}{\widetilde{m}_{B}^{-n} \widetilde{\gamma}_{B}^{n} \Gamma\left(\tilde{m}_{B}+n+1\right)}} \sum_{l=0}^{+\infty} \frac{\left(\begin{array}{c}
\tilde{m}_{B}+n \\
\end{array}\right)\left(2^{R}-1\right)^{l}\left(\widetilde{\bar{\gamma}}_{B} \widetilde{\bar{\gamma}}_{E}\right)^{n-l} \Gamma_{G}\left(\tilde{m}_{B}+ \tilde{m}_{E}+n-l\right)}{2^{l R}\left(2^{R} \tilde{m}_{B} \widetilde{\bar{\gamma}}_{E}+\tilde{m}_{E} \widetilde{\bar{\gamma}}_{B}\right)^{\widetilde{m}_{B}+\widetilde{m}_{E}+n-l}},
\end{align}
{\noindent}	 \rule[-6pt]{18cm}{0.05em}\\
\end{figure*}

We now focus on solving problem $(\ref{AO_method1})$. We will analyze the  outage probability constraint (\ref{out}) in detail. Recalling (\ref{AllError}), the rates of Bob and Eve $C_J(\boldsymbol{\omega}_I,{\bf{f}},\boldsymbol{\phi})$ are transformed as
\begin{align}
\label{rate}
C_J(\boldsymbol{\omega}_I,{\bf{f}},\boldsymbol{\phi})
= \log\left(1+\frac{\big |\boldsymbol{\phi}^H \left(\overline{\bf{G}}_{AJ}+\widetilde{\bf{G}}_{AJ} \right){\bf{f}}\Big |^2}{{\sigma}^2}\right), J\in\left\{ B,E\right\},
\end{align}
where $\overline{ \bf{G}}_{AJ}$ are determinate whose elements have the constant modulus $\frac{\kappa \sqrt{L_{AI}L_{IJ}}}{\kappa+1}$, while $\widetilde{\bf{G}}_{AJ}$ are random whose elements follow $\mathcal{CN}\left(0,\frac{2\kappa L_{AI}L_{IJ}}{\left(\kappa+1\right)^2}\right)$ based on proposition 1. Since the SR is monotonically increasing with transmit power, so the total power $P$ should be used, i.e., $\|{\bf{f}} \|^2=P$ at the optimal solution. Then we further have the following proposition 2.
\vspace{0ex}
\begin{prop}
If each element of $\widetilde{\bf{G}}_{AJ}, J\in\left\{ B,E\right\}$ follows $\mathcal{CN}\left(0,\frac{2\kappa L_{AI}L_{IJ}}{\left(\kappa+1\right)^2}\right)$ and transmit power is fixed, i.e., $\|{\bf{f}} \|^2=P$, $\boldsymbol{\phi}^H \widetilde{\bf{G}}_{AJ}{\bf{f}}$ follows $\mathcal{CN}\left({ 0},\frac{2\kappa L_{AI}L_{IJ}MP}{\left(\kappa+1\right)^2}\right)$ for any $\boldsymbol{\phi}$ and ${\bf f}$.
\end{prop}
\vspace{0ex}
\begin{proof}
See Appendix B.
\end{proof}
Proposition 2 shows that different $\boldsymbol{\phi}$ and ${\bf f}$ will not change the distribution of $\boldsymbol{\phi}^H \widetilde{\bf{G}}_{AJ}{\bf{f}}$. However, they still impact the value of $\boldsymbol{\phi}^H \overline{\bf{G}}_{AJ}{\bf{f}}$, which will impact the distributions of $C_J(\boldsymbol{\omega}_I,{\bf{f}},\boldsymbol{\phi}),J\in\left\{ B,E\right\}$ in  constraint (\ref{out}). Therefore, the challenge of solving problem $(\ref{AO_method1})$ still lies in how to handle constraint (\ref{out}).
Since our goal is to optimize the location  variable $\boldsymbol{\omega}_I$, we next further derive a unified upper bounds of constraint (\ref{out}), with which the outage probability of (\ref{out}) is only related to variable $\boldsymbol{\omega}_I$ but not to $\boldsymbol{\phi}$ and ${\bf f}$.
\begin{prop}
Outage probability ${\rm{Pr}}\left\{C_E(\boldsymbol{\omega}_I,{\bf{f}},\boldsymbol{\phi})\geq R_E \right\}$ could be upper bounded by
\begin{align}
\notag
 &{\rm{Pr}}\left\{C_E(\boldsymbol{\omega}_I,{\bf{f}},\boldsymbol{\phi})\geq R_E \right\}\\
 \label{outE}
 &\leq{\rm{Pr}}\left\{ \left(\frac{\kappa M\sqrt{N_tP}}{\kappa+1}+|\boldsymbol{\phi}^H\widetilde{\bf{G}}_{AE}^{small}{\bf{f}}|\right)^2\geq \frac{\left(2^{R_E}-1\right)\sigma^2}{L_{AI}L_{IE}} \right\}.
 \end{align}
  $\widetilde{\bf{G}}_{AE}^{small}\triangleq\frac{\sqrt{\kappa}}{\kappa+1}\left({\rm diag}\left({\bf{h}}_{IE}^{NLoS}\right){\bf{H}}_{AI}^{LoS} + {\rm diag}\left({\bf{h}}_{IE}^{LoS}\right){\bf{H}}_{AI}^{NLoS}\right)$, the elements of which follows $\mathcal{CN}\left({ 0},\frac{2\kappa M}{\left(\kappa+1\right)^2}\right)$ and any choice of $\boldsymbol{\phi}$ and ${\bf{f}}$ does not impact the distribution of $\boldsymbol{\phi}^H\widetilde{\bf{G}}_{AE}^{small}{\bf{f}}\sim \mathcal{CN}\left({ 0},\frac{2\kappa MP}{\left(\kappa+1\right)^2}\right)$.
\end{prop}
\begin{proof}
See Appendix C.
\end{proof}
For brevity, denote $\gamma_E \triangleq \frac{\kappa^2 M^2 N_tP}{\left(\kappa+1\right)^2}+|\boldsymbol{\phi}^H\widetilde{\bf{G}}_{AE}^{small}{\bf{f}}|^2+\frac{2\kappa M \sqrt{N_tP}}{\kappa+1}|\boldsymbol{\phi}^H\widetilde{\bf{G}}_{AE}^{small}{\bf{f}}|$.
With this proposition, $\gamma_E$ is a random variable as a function of random variable $\boldsymbol{\phi}^H\widetilde{\bf{G}}_{AE}^{small}{\bf{f}}$ and other constants. The distribution of $\boldsymbol{\phi}^H\widetilde{\bf{G}}_{AE}^{small}{\bf{f}}$ is irrelevant to $\boldsymbol{\phi}$ and ${\bf{f}}$, and so does the distribution of $\gamma_E$. Given $\kappa, M, N_t, P$, the cumulative distribution function of $\gamma_E$ is determinate, which could be denoted as $F_E(x)$. Then, we obtain $1-F\left(\frac{\left(2^{R_E}-1\right)\sigma^2}{L_{AI}L_{IE}}\right)\leq p_{out}$.
In fact, by denoting $\alpha_E \triangleq F^{-1}_E(1-p_{out})$, where $F^{-1}_E$ is the inverse cumulative distribution function, the constraint $(\ref{outE})\leq p_{out}$ can be equivalently transformed as $\alpha_E \leq \frac{\left(2^{R_E}-1\right)\sigma^2}{L_{AI}L_{IE}}$.
Note that if the constraint $\alpha_E \leq \frac{\left(2^{R_E}-1\right)\sigma^2}{L_{AI}L_{IE}}$ is satisfied, the original constraint ${\rm{Pr}}\left\{C_E(\boldsymbol{\omega}_I,{\bf{f}},\boldsymbol{\phi})\geq R_E \right\} \leq p_{out}$ in (\ref{out}) must be satified.  $\alpha_E$ could be calculated numerically in an off-line manner since $F_E(x)$ is determinate.

Next, we analyze the constraint ${\rm{Pr}}\left\{C_B(\boldsymbol{\omega}_I,{\bf{f}},\boldsymbol{\phi}) \leq R_B\right\}$ in (\ref{out}). We have proposition 4.
\begin{prop}
Outage probability ${\rm{Pr}}\left\{C_B(\boldsymbol{\omega}_I,{\bf{f}},\boldsymbol{\phi})\leq R_B \right\}$ could be upper bounded by
\begin{align}
\notag
 &{\rm{Pr}}\left\{C_B(\boldsymbol{\omega}_I,{\bf{f}},\boldsymbol{\phi})\leq R_B \right\}\\
 \label{outB}
 & \leq {\rm{Pr}}\left\{ \left( \frac{\kappa M\sqrt{N_tP}}{\kappa+1}-|\boldsymbol{\phi}^H\widetilde{\bf{G}}_{AB}^{small}{\bf{f}}| \right)^2\leq \frac{\left(2^{R_B}-1\right)\sigma^2}{L_{AI}L_{IB}} \right\}.
 \end{align}
 $\widetilde{\bf{G}}_{AB}^{small}\triangleq\frac{\sqrt{\kappa}}{\kappa+1}\left({\rm diag}\left({\bf{h}}_{IB}^{NLoS}\right){\bf{H}}_{AI}^{LoS} + {\rm diag}\left({\bf{h}}_{IB}^{LoS}\right){\bf{H}}_{AI}^{NLoS}\right)$, the elements of which follows $\mathcal{CN}\left({ 0},\frac{2\kappa M}{\left(\kappa+1\right)^2}\right)$ and any choice of $\boldsymbol{\phi}$ and ${\bf{f}}$ does not impact the distribution of $\boldsymbol{\phi}^H\widetilde{\bf{G}}_{AB}^{small}{\bf{f}}\sim \mathcal{CN}\left({ 0},\frac{2\kappa MP}{\left(\kappa+1\right)^2}\right)$.
\end{prop}
\vspace{0ex}
\begin{proof}
 See Appendix D.
\end{proof}

Similarly, we denote $\gamma_B \triangleq \frac{\kappa^2 M^2 N_tP}{\left(\kappa+1\right)^2}+|\boldsymbol{\phi}^H\widetilde{\bf{G}}_{AB}^{small}{\bf{f}}|^2-\frac{2\kappa M \sqrt{N_tP}}{\kappa+1}|\boldsymbol{\phi}^H\widetilde{\bf{G}}_{AB}^{small}{\bf{f}}|$,
 denote the cumulative distribution function of $\gamma_B$ as $F_B(x)$, and $\alpha_B \triangleq F^{-1}_B(p_{out})$. It is not hard to see that
 the constraint $(\ref{outB})\leq p_{out}$ can be equivalently transformed as $\alpha_B \geq \frac{\left(2^{R_B}-1\right)\sigma^2}{L_{AI}L_{IB}}$, which could be calculated numerically in an off-line manner once $p_{out}$ is given. Note that if this constraint holds, the original constraint ${\rm{Pr}}\left\{C_B(\boldsymbol{\omega}_I,{\bf{f}},\boldsymbol{\phi})\leq R_B \right\} \leq p_{out}$ in (\ref{out}) must hold.

  After transforming (\ref{out}) as $\alpha_E\leq \frac{\left(2^{R_E}-1\right)\sigma^2}{L_{AI}L_{IE}}$ and $\alpha_B\geq \frac{\left(2^{R_B}-1\right)\sigma^2}{L_{AI}L_{IB}}$, we have following proposition.

\begin{prop}
(\ref{AO_method1}) can be equivalently transformed as
 \begin{align}
 \label{max_location}
 \max \limits_{\boldsymbol{\omega}_I} \ \frac{\sigma^2+\alpha_B L_{AI}L_{IB}}{\sigma^2+\alpha_E L_{AI}L_{IE}}, \ s.t.\ \boldsymbol{\omega}_I\in {\Omega}_I,
 \end{align}
where $L_{AI}=\frac{L_0 }{\|\boldsymbol{\omega}_{I}-\boldsymbol{\omega}_{A}\|^{\rho_{AI}}}$, $L_{IB}=\frac{L_0}{\|\boldsymbol{\omega}_{I}-\boldsymbol{\omega}_{B}\|^{\rho}}$ and $ L_{IE}=\frac{L_0}{\|\boldsymbol{\omega}_{I}-\boldsymbol{\omega}_{E}\|^{\rho}}$ are functions of $\boldsymbol{\omega}_{I}$.
\end{prop}
\begin{proof}
 See Appendix E.
\end{proof}

Problem (\ref{max_location}) is non-convex due to the objective function. To solve it, we can globally search over $\boldsymbol{\omega}_I$ to obtain the optimal solution. However the computation complexity for this method is high, so we propose a successive convex approximation (SCA) method with a low complexity to solve (\ref{max_location}) and take this global search method as a benchmark.
 First we find a lower bound $\frac{\alpha_BL_{AI}L_{IB}}{\sigma^2+\alpha_EL_{AI}L_{IE}}$ of the objective function by ignoring the noise in numerator and this lower bound is approximately equal to the original objective function because the path loss much larger than noise by appropriately optimizing the location of IRS. Then, (\ref{max_location}) is transformed as
\begin{align}
\label{min_location}
\min \limits_{\boldsymbol{\omega}_I} \ \frac{\sigma^2+\alpha_EL_{AI}L_{IE}}{\alpha_BL_{AI}L_{IB}}, \ {\rm{s.t.}}\ \boldsymbol{\omega}_I \in \Omega_{I}.
\end{align}
To further solve the problem, we introduce auxiliary variables ${\bf a}=[a_{AI},a_{IB},a_{IE},a_{AB},a_{BE}]^T$,
where $ a_{AI}=\frac{1}{L_{AI}}$, $ a_{IE}=\frac{1}{L_{IE}}=$, $ a_{IB}=\frac{1}{L_{IB}}$, $a_{AB}= a_{AI}a_{IE}$, $a_{BE}= \frac{a_{IB}}{a_{IE}}$, and transform
(\ref{min_location}) as
\begin{subequations}\label{auxiliary}
 \begin{align}
  \min \limits_{\boldsymbol{\omega}_I,{\bf a}} &\ \frac{\sigma^2}{\alpha_B}a_{AB}+\frac{\alpha_E}{\alpha_B}a_{BE}, \label{auxiliarya}\\
 {\rm{s.t.}} &\
 a_{AI}\geq \frac{\|\boldsymbol{\omega}_{I}-\boldsymbol{\omega}_{A}\|^{\rho_{AI}}}{L_0}, \
 a_{IB}\geq \frac{\|\boldsymbol{\omega}_{I}-\boldsymbol{\omega}_{B}\|^{\rho}}{L_0},\label{auxiliaryb}\\
 & a_{IE} \leq \frac{\|\boldsymbol{\omega}_{I}-\boldsymbol{\omega}_{E}\|^{\rho}}{L_0}, \ a_{AB}\geq a_{AI}a_{IB},\ a_{BE}\geq \frac{a_{IB}}{a_{IE}}. \label{auxiliaryc}
\end{align}

Note that in problem (\ref{auxiliary}), the inequalities in constraints (\ref{auxiliaryb}) and (\ref{auxiliaryc}) will be active at the optimum. This could be proved by the contradiction. Assume that at the optimum one of the corresponding constraints in (\ref{auxiliary}) is a strict inequality. Then, we can always decrease $a_{AI}$ to satisfy the constraint with equality, which decreases the objective value. Therefore, for the optimal solution of (\ref{auxiliary}), $a_{AI}$ must be satisfied with equality. Similarly, all the other auxiliary variables must be active at the optimum. Now, the objective function and constraint (\ref{auxiliaryb}) are convex but constraints in (\ref{auxiliaryc}) are still non-convex. In the following, we propose an SCA method \cite{T-16} to solve (\ref{auxiliary}) iteratively by exploiting the first-order Taylor expansion of all the constraints.

First, by given a feasible point $\boldsymbol{\omega}_{I}^{(l)}$, the upper bound for $a_{IE} \leq \frac{\|\boldsymbol{\omega}_{I}-\boldsymbol{\omega}_{E}\|^{\rho}}{L_0}$ in (\ref{auxiliaryc}) is given by
 \begin{align}
\notag
 L_0a_{IE}&\leq \| \boldsymbol{\omega}_{I}^{(l)}-\boldsymbol{\omega}_{E}\|^{\rho}+\rho \left( \| \boldsymbol{\omega}_{I}^{(l)}-\boldsymbol{\omega}_{E}\|^2\right)^{\frac{\rho}{2}-1}\\
  \label{taylor}
 &\times\left( \boldsymbol{\omega}_{I}^{(l)}-\boldsymbol{\omega}_{E}\right)^T\left(\boldsymbol{\omega}_{I}-\boldsymbol{\omega}_{I}^{(l)} \right),
\end{align}
Next, we equivalently transform the constraints $ a_{AB}\geq a_{AI}a_{IB}$ and $\ a_{BE}\geq \frac{a_{IB}}{a_{IE}}$ as
 \begin{align}
 \notag
 &\ a_{AB}\geq a_{AI}a_{IB}=\frac{1}{2}\left[\left(a_{AI}+a_{IB} \right)^2-a_{AI}^2-a_{IB}^2\right],\\
 \notag
 &\ a_{BE}\geq \frac{a_{IB}}{a_{IE}}\Leftrightarrow\\
 & a_{IB}\leq a_{IE}a_{BE}=\frac{1}{2}\left[\left(a_{IE}+a_{BE} \right)^2-a_{IE}^2-a_{BE}^2\right].
\end{align}
Then, the convex upper and lower bounds at given  points $a_{AI}^{(l)},a_{IE}^{(l)},a_{IB}^{(l)}$ and $a_{BE}^{(l)}$ are given by (\ref{taylor1}) and (\ref{taylor2}) at the top of next page.

\end{subequations}
\begin{figure*}
\begin{align}
  \label{taylor1}
  a_{AI}a_{IB}\leq b_1 &\triangleq\frac{1}{2}\left[\left(a_{AI}+a_{IB} \right)^2-\left({a_{AI}^{(l)}}^2+2a_{AI}^{(l)}\left(a_{AI}-a_{AI}^{(l)} \right)\right)-\left({a_{IB}^{(l)}}^2+2a_{IB}^{(l)}\left(a_{IB}-a_{IB}^{(l)} \right)\right)\right],\\
 \label{taylor2}
  a_{IE}a_{BE}\geq b_2&\triangleq \frac{1}{2}\left(a_{IE}^{(l)}+a_{BE}^{(l)} \right)^2+\left(a_{IE}^{(l)}+a_{BE}^{(l)} \right)\left(a_{IB}-a_{IB}^{(l)} \right)
  +\left(a_{IE}^{(l)}+a_{BE}^{(l)} \right)\left(a_{BE}-a_{BE}^{(l)} \right)-\frac{1}{2}(a_{IE}^2+a_{BE}^2).
 \end{align}
{\noindent}	 \rule[-6pt]{18cm}{0.05em}\\
\end{figure*}

 Finally, the problem $(\ref{auxiliary})$ is transformed as
 \begin{align}
\notag
 \min \limits_{\boldsymbol{\omega}_I,{\bf a}} &\ \frac{\sigma^2}{\alpha_B}a_{AB}+\frac{\alpha_E}{\alpha_B}a_{BE}, \\
  \label{auxiliary_final}
  {\rm{s.t.}} &\ (\ref{auxiliaryb}),\ (\ref{taylor}), \ a_{AB}\geq b_1, \ a_{IB}\leq b_2,
\end{align}
where all the constraints are convex, so this problem can be conveniently solved. The overall algorithm to maximize $R$ by optimizing  $\boldsymbol{\omega}_I$ is summarized in Algorithm 1.
\begin{algorithm}[h]
	\caption{SCA method To Solve Problem (\ref{auxiliary_final})}
	\begin{algorithmic}
         \State 1. set the convergence precision $\epsilon$ and $\boldsymbol{\omega}_I^{(0)}$, $l=0$;
         \Repeat
         \State 2. solve problem \ref{auxiliary_final}) to obtain $\boldsymbol{\omega}_I^{(l)}$;
         \State 3. $l=l+1$;
         \Until{$\|\boldsymbol{\omega}_I^{(l+1)}-\boldsymbol{\omega}_I^{(l)}\|^2\leq \epsilon$;}
         \State 4. output $\boldsymbol{\omega}_I$.
	\end{algorithmic}
\end{algorithm}

\subsection{ Optimize Beamformer and Phase Shift}\label{BeamPhase}
After obtain the optimized location of IRS, the next step is to optimize the beamformer ${\bf{f}}$ and phase shifts $\boldsymbol{\phi}$. Note that once  a deployed location is fixed, the instantaneous CSI of Alice-IRS and IRS-Bob are available by channel training, and LoS components of IRS-Eve is also available, but the NLoS components of IRS-Eve link are unknown random variables. The CSI model of Eve now is expressed as (\ref{EveError}). The sub-problem to optimize $\left({\bf{f}},\boldsymbol{\phi}\right)$ is expressed as
\begin{subequations}\label{AO}
 \begin{align}
\max \limits_{{\bf{f}},\boldsymbol{\phi}} &\ R, \label{AOa} \\
 {\rm{s.t.}} \  &\ {\rm{Pr}}\left\{R_s({\bf{f}},\boldsymbol{\phi})\geq R\right\} \geq 1-p_{out}, \label{AOb}\\
&\ \|{\bf{f}}\|^2 \leq P, |\boldsymbol{\phi}_{i}|=1,\ i=1,2,\cdots,M,\label{AOc}
\end{align}
\end{subequations}
This problem is non-convex due to the outage probability constraint (\ref{AOb}) so we first handle this constraint. To make it tractable, we rewrite the inequality $R_s({\bf{f}},\boldsymbol{\phi})\geq R$ as
 \begin{align}
\notag
 &\log\left(1+\frac{|\boldsymbol{\phi}^H{\bf{G}}_{AB}{\bf{f}}|^2}{{\sigma}^2}\right) - \log\left(1+\frac{|\boldsymbol{\phi}^H
 {\bf{G}}_{AE}{\bf{f}}|^2}{{\sigma}^2}\right)\geq R,\\
 \label{Trans}
 &\Leftrightarrow {\rm Tr}\left({{\bf{G}}}_{AE}{\bf{F}}{{\bf{G}}}_{AE}^H{\bf{Q}}\right)\leq 2^{-R}\left({\sigma}^2+
{\rm Tr}\left({{\bf{G}}}_{AB}{\bf{F}}{{\bf{G}}}_{AB}^H{\bf{Q}}\right)\right)-{\sigma}^2,
 \end{align}
where ${\bf{F}}\triangleq {\bf{f}}{\bf{f}}^H$ and ${\bf{Q}}\triangleq \boldsymbol{\phi}\boldsymbol{\phi}^H$. Here, ${\bf G}_{AB}$ is known, ${\bf G}_{AE}=\overline{\bf G}_{AE} + \widetilde{\bf G}_{AE}$ is modeled in (\ref{EveError}) with determinate and random components $\overline{\bf G}_{AE}$ and $\widetilde{\bf G}_{AE}$, where each element of $\widetilde{{\bf{G}}}_{AE} $ follows $\mathcal{CN}\left({0},\frac{\kappa L_{AI}L_{IE}} {(\kappa +1)^2}\right)$ based on proposition 1. By substituting ${\bf G}_{AE}$ into (\ref{Trans}), we could further obtain
\begin{align}
\notag
 &{\rm Tr}\left(\left(\overline{\bf{G}}_{AE}+\widetilde{\bf G}_{AE}\right){\bf{F}}\left(\overline{\bf{G}}_{AE}^H
 +\widetilde{\bf G}_{AE}^H\right){\bf{Q}}\right)\\
 \notag
 &\leq 2^{-R}\left({\sigma}^2+
{\rm Tr}\left({{\bf{G}}}_{AB}{\bf{F}}{{\bf{G}}}_{AB}^H{\bf{Q}}\right)\right)-{\sigma}^2,\\
\label{trans}
 \Leftrightarrow
  &\underbrace{{\rm Tr}\left(\widetilde{\bf G}_{AE}{\bf{F}}\widetilde{\bf G}_{AE}^H{\bf{Q}}\right)}_{f_1}+
 2{\rm{Re}}\underbrace{\left\{{\rm Tr}\left(\overline{{{\bf{G}}}}_{AE}{\bf{F}}\widetilde{\bf G}_{AE}^H{\bf{Q}}\right)\right\}}_{f_2}\leq c_1,
 \end{align}
 where $c_1 \triangleq2^{-R}\left({\sigma}^2+
{\rm Tr}\left({{\bf{G}}}_{AB}{\bf{F}}{{\bf{G}}}_{AB}^H{\bf{Q}}\right)\right)-{\sigma}^2-{\rm Tr}\left( \overline{\bf{G}}_{AE}{\bf{F}}\overline{\bf{G}}_{AE}^H{\bf{Q}}\right)$ does not include the random variable $\widetilde{{\bf{G}}}_{AE}$, but both $f_1$ and $f_2$ include it.

To handle $f_1$ and $f_2$, we first denote ${\bf{g}}_{AE}\triangleq{\rm vec}\left( \widetilde{{\bf{G}}}_{AE}\right) \in \mathbb{C}^{MN_t\times 1} \sim \mathcal{CN}\left({\bf{0}},\delta_{AE}^2{\bf{I}}_{MN_t}\right)$ as ${\bf{g}}_{AE}= \delta_{AE}
{\bf{u}}$, where $\delta_{AE}=\sqrt{\frac{\kappa L_{AI}L_{IE}} {(\kappa +1)^2}}$, and ${\bf{u}} \in \mathbb{C}^{MN_t\times 1}\sim \mathcal{CN}\left({\bf0},{{\bf{I}}}_{MN_t}\right)$. Then $f_1$ in $(\ref{trans})$ can be reformulated as
\begin{align}
 \label{f1}
f_1
\overset{(c)}{=} {\bf{g}}_{AE}^H\left( {\bf{F}}^T\otimes {\bf{Q}}\right){\bf{g}}_{AE}
 =\delta_{AE}^2 {\bf{u}}^H\left( {\bf{F}}^T\otimes {\bf{Q}}\right){\bf{u}}
 \triangleq {\bf{u}}^H{\bf{A}}_{AE}{\bf{u}},
\end{align}
where $(c)$ is obtained by invoking the identity ${\rm Tr}\left({\bf{A}}^H\bf{BCD}\right)={\rm vec}^H\left({\bf{A}}\right)\left( \bf{D}^T\otimes {\bf{B}}\right)\rm vec(\bf{C})$ and ${\bf{A}}_{AE}=\delta_{AE}^2 \left( {\bf{F}}^T\otimes {\bf{Q}}\right)$. Similarly, the expression $f_2$ in $(\ref{trans})$ can be reformulated as
\begin{align}
\label{f2}
f_2  =& \delta_{AE}{\bf{u}}^H\left( {\bf{F}}^T\otimes {\bf{Q}}\right){\rm vec}\left(\overline{\bf{G}}_{AE}\right)
 \triangleq {\bf{u}}^H{\bf{a}}_{AE},
\end{align}
where ${\bf{a}}_{AE}=\delta_{AE}\left( {\bf{F}}^T\otimes {\bf{Q}}\right){\rm vec}\left(\overline{\bf{G}}_{AE}\right)$. By substituting $(\ref{f1})$ and$(\ref{f2})$ into $(\ref{trans})$, we have
\begin{align}
(\ref{trans}) \Leftrightarrow  {\bf{u}}^H{\bf{A}}_{AE}{\bf{u}}+2{\rm{Re}}\left\{{\bf{u}}^H{\bf{a}}_{AE}\right\}-c_1\leq 0,
\end{align}
and the outage constraint $(\ref{AOb})$ can be reformulated as
\begin{align}
\label{pout_final}
{\rm{Pr}}\left\{{\bf{u}}^H{\bf{A}}_{AE}{\bf{u}}+2{\rm{Re}}\left\{{\bf{u}}^H{\bf{a}}_{AE}\right\}-c_1\geq 0\right\} \leq p_{out}.
\end{align}
With the probability constraint in form of $(\ref{pout_final})$, we could exploit Bernstein-Type Inequality-I (BTI-I) in Lemma 1 to handle it which has quadratic forms of Gaussian variables matrix.
\begin{lemma}
(Bernstein-Type Inequality-I \cite{Ma-14}): Let ${\bf{G}} = {\bf{x}}^H{\bf{Cx}}+2{\rm Re}\left\{{\bf{x}}^H{\bf{c}}\right\}$, where ${\bf{C}}\in \mathbb{C}^{N\times N}$
is a complex Hermitian matrix, ${\bf{c}} \in \mathbb{C}^{N\times 1}$, and ${\bf{x}}\sim \mathcal{CN}(\bf{0},{{\bf{I}}})$. Then for any $\varrho \geq 0$, we have
\begin{align}
\notag
 {\rm{Pr}} \left\{{\bf G}\geq \rm Tr({\bf C}) +\sqrt{2\varrho}\sqrt{||\rm vec({\bf C})||^2+2||{\bf c}||^2}+\varrho \lambda^+({\bf C})\right\} \leq e^{-\varrho},
\end{align}
where $\lambda^+({\bf C})=max\left\{ \lambda_{max}({\bf C}),0\right\}$, and $\lambda_{max}({\bf C})$ represents the maximum eigenvalue of ${\bf C}$.
\end{lemma}
With BTI-I, $(\ref{pout_final})$ is transformed to a deterministic form as
\begin{align}
\notag
&{\rm Tr}({\bf{A}}_{AE}) +\sqrt{2\varrho}\sqrt{\|{\rm vec}({\bf{A}}_{AE})\|^2++2\|{\bf{a}}_{AE}\|^2}\\
\label{BTI}
&+\varrho \lambda^+({\bf{A}}_{AE})-c_1\leq 0,
\end{align}
where $\varrho = -{\rm ln}(p_{out})$. If $(\ref{BTI})$ is true, $(\ref{pout_final})$ must hold true. Consequently, $(\ref{AO})$ is transformed as
\begin{subequations}\label{ NextPerfectLocation}
 \begin{align}
 \max \limits_{{\bf{F}} ,{\bf{Q}}} &\ R, \label{ NextPerfectLocationa} \\
 \notag
 \ {\rm{s.t.}} & \ (\ref{BTI}),\ {\rm Tr}({\bf{F}})\leq P,\ {\rm rank}({\bf{F}})=1,{\bf F} \succeq {\bf 0}, \\
 \ &\ {\rm Diag}\left({\bf{Q}}\right)= {\bf{1}}_M,\ {\rm rank}({\bf{Q}})=1,\ {\bf{Q}} \succeq {\bf 0} \label{ NextPerfectLocationc},
\end{align}
\end{subequations}
where ${\rm Diag}\left({\bf{Q}}\right)= {\bf{1}}_M$ ensures the phase shifts $\boldsymbol{\phi}$ with unit modulus.
 Although ${\bf{F}}$ and ${\bf{Q}}$ are coupled in (\ref{BTI}), fortunately, this constraint is convex for ${\bf{F}}$ given ${\bf{Q}}$, and is convex for ${\bf{Q}}$ given ${\bf{F}}$ as well, respectively. All the other constraints by dropping rank-one constraint in $(\ref{ NextPerfectLocation})$ are also convex. Hence, an AO algorithm is developed to optimize ${\bf{F}}$ and ${\bf{Q}}$ iteratively.

\subsubsection{Optimize Beamformer}\label{Beam}
Firstly, we optimize ${\bf{F}}$ by fixed ${\bf{Q}}$, which is expressed as
\begin{align}
\label{beam}
\max \limits_{{\bf{F}}\succeq {\bf 0}} \ R, \ {\rm{s.t.}} \ (\ref{BTI}), \ {\rm Tr}({\bf{F}})\leq P, \ {\rm rank}({\bf{F}})=1.
\end{align}
In fact, to maximize $R$ in (\ref{beam}) is equivalently to first solve a power minimization (PM) problem and then take a bisection search over $R$ to obtain the optimal $R^*$ because the optimal value of PM problem is monotonically increasing with respect to $R$ \cite{Ma-14}. Thus, solving PM problem with different $R$ and using a bisection search over $R$, $R^{*}$ can be obtained.

For given a target SR $R>0$, PM problem is expressed as
\begin{align}
  \label{beam2} \min \limits_{{\bf{F}}} \ {\rm Tr}({\bf F}), \ {\rm{s.t.}} \ (\ref{BTI}), \ {\rm rank}({\bf{F}})=1,\  {\bf{F}}\succeq {\bf 0}.
\end{align}
(\ref{beam2}) is equivalent to the following relaxed problem by dropping the constraint ${\rm{rank}}({\bf{F}})=1$.
\begin{subequations}\label{beam4}
\begin{align}
 \min \limits_{{\bf{F}}\succeq {\bf 0},\zeta,\upsilon} &\ {\rm Tr}\left({\bf{F}}\right) \label{BTIa}\\
 \notag
{\rm{s.t.}} &\ {\rm Tr}\left({\bf{A}}_{AE}\right) +\sqrt{2\varrho}\zeta+\varrho \upsilon+{\sigma}^2+ {\rm Tr}\left( \overline{\bf{G}}_{AE}{\bf{F}}\overline{\bf{G}}_{AE}^H{\bf{Q}}\right)\\
\notag
&-2^{-R}\left({\sigma}^2+
{\rm Tr}\left({{\bf{G}}}_{AB}{\bf{F}}{{\bf{G}}}_{AB}^H{\bf{Q}}\right)\right)\leq 0,\\
&\ \begin{Vmatrix}
 {\rm vec}\left({\bf{A}}_{AE}\right)\\
2{\bf{a}}_{AE}
\end{Vmatrix}\leq \zeta, \  \upsilon{{\bf{I}}}-{\bf{A}}_{AE}\succeq {\bf{0}},\ \upsilon\geq 0, \label{BTId}
\end{align}
\end{subequations}
where $\upsilon$ and $\zeta$ are the slack variables. Since ${\bf{A}}_{AE}$ and ${\bf{a}}_{AE}$ are liner for ${\bf F}$, so
 ${\rm Tr}\left({\bf{A}}_{AE}\right)$, ${\rm Tr}\left({{\bf{G}}}_{AB}{\bf{F}}{{\bf{G}}}_{AB}^H{\bf{Q}}\right)$ and ${\rm Tr}\left( \overline{\bf{G}}_{AE}{\bf{F}}\overline{\bf{G}}_{AE}^H{\bf{Q}}\right)$ are liner constraints, $\begin{Vmatrix}
 {\rm vec}\left({\bf{A}}_{AE}\right)\\
2{\bf{a}}_{AE}
\end{Vmatrix}$ is a second cone (SOC) constraint and $\upsilon{{\bf{I}}}-{\bf{A}}_{AE}\succeq {\bf{0}}$, ${\bf{F}}\succeq {\bf{0}}$ are linear matrix inequality (LMI) constraints. Therefore, it is a convex problem which can be solved \cite{Boyd-04}. Considering the rank-1 constraint, in the following proposition we prove that we can always obtain a rank-one optimal ${\bf{F}}$ if the problem is feasible.
Thus, the optimal ${\bf{f}}$ can be obtained by eigen-decomposition of ${\bf{F}}$.

\begin{prop}
A rank-one solution in $(\ref{beam4})$ can always be obtained if $(\ref{beam4})$ is feasible.
\end{prop}
\begin{proof}
Please refer to Appendix F.
\end{proof}
The overall algorithm to maximize $R$ by optimizing ${\bf F}$ is summarized in Algorithm 1.
\begin{algorithm}[h]
	\caption{Bisection Method To Solve Problem $(\ref{beam})$}
	\begin{algorithmic}
         \State 1. set $\epsilon=10^{-3}$, the upper $\overline{R}_u$ and lower bound $\overline{R}_l$;
         \State 2. let $\overline{R}_{mid}=\frac{\overline{R}_u+\overline{R}_l}{2}$ and solve (\ref{beam4}) with $R=\overline{R}_{mid}$;
         \State 3. if (\ref{beam4}) is feasible, check the power constraint ${\rm Tr}({\bf F})\leq P$, if power constraint satisfies, set $\overline{R}_l = \overline{R}_{mid}$ and go to step 5; otherwise, set $\overline{R}_u = \overline{R}_{mid}$ and go to step 2;
         \State 4. if (\ref{beam4}) is infeasible, let $\overline{R}_u = \overline{R}_{mid}$ and go to step 2;
         \State 5. if $\overline{R}_u-\overline{R}_l\leq \epsilon$, stop and $R^*=\overline{R}_l$; else go to step 2.
	\end{algorithmic}
\end{algorithm}

\subsubsection{Optimize Phase Shift}\label{Phase}
The sub-problem of optimizing ${\bf Q}$ under fixed ${\bf{F}}$  is expressed as
\begin{align}
\label{shift}
\max \limits_{ {\bf{Q}}\succeq {\bf 0}} \ R,\
 {\rm{s.t.}} \ (\ref{BTI}), \ {\rm Diag}\left({\bf{Q}}\right)= {\bf{1}}_M,\ {\rm rank}({\bf{Q}})=1.
\end{align}
The optimal value of (\ref{shift}) can be obtained by first solving a feasibility check problem and then using bisection  search.
For given a target $R>0$, the feasibility check problem is
\begin{align}
\label{shift2}
\mathop{\rm{Find}}\limits_{{\bf Q}\succeq {\bf 0}}\ {\bf Q}, \ {\rm{s.t.}}\ (\ref{BTI}), \ {\rm Diag}\left({\bf{Q}}\right)= {\bf{1}}_M,\ {\rm rank}({\bf{Q}})=1.
\end{align}
Obviously, the feasible area is monotonically decreasing with respect to $R$. Thus, an iterative bisection search over $R$ can obtain $R^{*}$. To solve (\ref{shift2}), we first reformulate it as (\ref{shift3}) by dropping the rank one constraint.
\begin{subequations}\label{shift3}
\begin{align}
 \mathop{\rm{Find}}\limits_{{\bf Q}\succeq {\bf 0}} &\  {\bf{Q}} \label{finda}\\
 \notag
{\rm{s.t.}} \ &\  {\rm Tr}\left({\bf{A}}_{AE}\right) +\sqrt{2\varrho}\alpha+\varrho \beta+{\sigma}^2+ {\rm Tr}\left( \overline{\bf{G}}_{AE}{\bf{F}}\overline{\bf{G}}_{AE}^H{\bf{Q}}\right)\\
\notag
&-2^{-R}\left({\sigma}^2+
{\rm Tr}\left({{\bf{G}}}_{AB}{\bf{F}}{{\bf{G}}}_{AB}^H{\bf{Q}}\right)\right)\leq 0, \ {\rm Diag}\left({\bf{Q}}\right)= {\bf{1}}_M,\\
 & \begin{Vmatrix}
 {\rm vec}\left({\bf{A}}_{AE}\right)\\
 2{\bf{a}}_{AE}
\end{Vmatrix}\leq \alpha,  \  \beta{{\bf{I}}}-{\bf{A}}_{AE}\succeq {\bf{0}}, \ \beta\geq 0,  \label{findd}
\end{align}
\end{subequations}
If dropping the rank-one constraint of ${\rm{rank}}({\bf{Q}})=1$, problem $(\ref{shift3})$ is convex with the same analysis of $(\ref{beam4})$. Therefore, we use SDR method to find a solution and then use sequential rank-one constraint relaxation algorithm to recover the rank-one constraint. In such a way, $\boldsymbol{\phi}$ is a KKT stationary solution
of (\ref{shift}) \cite{Cao-17}.
The algorithm is summarized in Algorithm 2.
\begin{algorithm}[h]
	\caption{Bisection Method To Solve Problem $(\ref{shift})$}
	\begin{algorithmic}
         \State 1. set $\epsilon=10^{-3}$, the upper $\overline{R}_u$ and lower bound $\overline{R}_l$ of $R$;
         \Repeat
         \State 2. let $\overline{R}_{mid}=\frac{\overline{R}_u+\overline{R}_l}{2}$ and solve $(\ref{shift3})$ with $R=\overline{R}_{mid}$;
         \State 3. if infeasible, let $\overline{R}_u = \overline{R}_{mid}$; otherwise, let $\overline{R}_l = \overline{R}_{mid}$;
         \Until{$\overline{R}_u-\overline{R}_l\leq \epsilon$};
         \State 4. output the optimal the $R^*$ and recover rank-one constraint to obtain the local optimal $\boldsymbol{\phi}^*$;
	\end{algorithmic}
\end{algorithm}
\subsection{Algorithm Analysis}
 The overall two-stage algorithm to solve the joint optimization problem (\ref{ PerfectLocation}) of IRS location, beamformer and reflection coefficient  is summarized in Algorithm 3.
\begin{algorithm}[h]
	\caption{Two-Stage Method To Solve Problem (\ref{ PerfectLocation})}
	\begin{algorithmic}
         \State 1. set the convergence precision $\epsilon=10^{-3}$, initial $\boldsymbol{\phi}^{(0)}$, $n=0$;
         \State 2. solve problem (\ref{auxiliary_final}) to obtain the sub-optimal $\boldsymbol{\omega}_I$;
         \Repeat
         \State 3. solve $(\ref{beam})$ with Algorithm 1 to obtain ${\bf f}^{(n+1)}$ ;
         \State 4. solve $(\ref{shift})$ with Algorithm 2 to obtain ${\boldsymbol{\phi}}^{(n+1)}$ and $R^{(n+1)}$;
         \State 5. $n=n+1$;
         \Until{$R^{(n+1)}-R^{(n)}\leq \epsilon$;}
	\end{algorithmic}
\end{algorithm}

 Based on the algorithm description, we provide a brief analysis on the convergence and complexity. For optimizing $\boldsymbol{\omega}_I$ in $(\ref{auxiliary})$, this problem only involve vector multiplication so the computational
 complexity is $\mathcal{O}\left\{ n\right\}$. For given IRS location, by iteratively solving  $(\ref{beam})$ with optimal and solving $(\ref{shift})$ with local optimal, the target SR $R$ can be monotonically increased with guaranteed convergence. Note that for given $\boldsymbol{\omega}_I$, $R\left({{\bf{f}}^{(n)}},\boldsymbol{\phi}^{(n)}\right)$ is a feasible solution for $(\ref{ NextPerfectLocation})$ in the $n$-th iteration. For the next iteration, we solve $(\ref{beam})$ optimally by fixed $\boldsymbol{\phi}^{(n)}$
so that $R\left({{\bf{f}}^{(n)}},\boldsymbol{\phi}^{(n)}\right)\leq R\left({{\bf{f}}^{(n+1)}},\boldsymbol{\phi}^{(n)}\right)$. Then, a local optimal solution
$\boldsymbol{\phi}^{(n+1)}$ can be obtained by solve $(\ref{shift})$ under fixed ${{\bf{f}}^{(n+1)}}$, so that
$R\left({{\bf{f}}^{(n+1)}},\boldsymbol{\phi}^{(n)}\right)\leq R\left({{\bf{f}}^{(n+1)}},\boldsymbol{\phi}^{(n+1)}\right)$. Therefore, the two iterative steps
follow $R\left({{\bf{f}}^{(n)}},\boldsymbol{\phi}^{(n)}\right)\leq R\left({{\bf{f}}^{(n+1)}},\boldsymbol{\phi}^{(n+1)}\right)$ thus guarantee the convergence.

 Since both the resulting convex problem $(\ref{beam4})$ and $(\ref{shift3})$ involve two LMI, one SOC and there linear constraints that can be solved
 by a standard interior point method, the general expression for computational complexity has been given in \cite{Zhou-20}. The computational
 complexities of proposed method for (\ref{beam4}) and (\ref{shift3}) in per iteration are listed in Table I.
\begin{table}[h]
		\centering
      \caption{computational complexities}
		\label{tab:Margin_settings}
		\begin{tabular}{|l|c|}\hline
            \multirow{2}{*}{(\ref{beam4})}&$\mathcal{O}\left\{\left(N_tM+N_t+2\right)^{1 / 2} N_t^2\left[N_t^{4}+N_t^2\left(N_t^2M^2+N_t^2\right)\right. \right.$\\
               &$\left. \left. +N_t^3M^3+N_t^3+N_t^2\left(N_t^2M^2+N_tM\right)^2\right]\right\}$ \\ \hline
               \multirow{2}{*}{(\ref{shift3})}&$\mathcal{O}\left\{\left(N_tM+M+2\right)^{1 / 2} M^2\left[M^4+M^2\left(N_t^2M^2+M^2\right)\right. \right.$ \\
               &$\left. \left. +N_t^3M^3+M^3+M^2 \left(N_t^2M^2+N_tM\right)^2\right]\right\}$ \\ \hline
		\end{tabular}
\end{table}

\section{ Problem Solution With Location Region Of Eve}
In this section, we aim to maximize SR without exact location of Eve but only a suspicious area where an Eve may exist. In this case, both large-scale and small-scale fadings of IRS-Eve link are unknown, which prevents the optimization of IRS location. To guarantee the security, we consider the worst case that the information leakage to Eve is maximum, which means we will maximize the minimum secrecy rate. The problem is expressed as
\begin{subequations}\label{worst_Eve}
\begin{align}
\max \limits_{{\bf{f}},\boldsymbol{\phi},\boldsymbol{\omega}_I} \min \limits_{\boldsymbol{\omega}_E}&\ R \label{worst_Evea}\\
 {\rm{s.t.}} &\ {\rm{Pr}}\left\{R_s(\boldsymbol{\omega}_I,\boldsymbol{\omega}_E,{\bf{f}},\boldsymbol{\phi})\leq R\right\} \leq p_{out},\label{worst_Eveb}\\
 \notag
&\ \|{\bf{f}}\|^2 \leq P, |\boldsymbol{\phi}_{i}|=1,\ i=1,\cdots,M,\\
& \boldsymbol{\omega}_I\in \Omega_I, \ \boldsymbol{\omega}_E\in \Omega_E.\label{worst_Evec}
\end{align}
\end{subequations}

Since  $\boldsymbol{\omega}_I$ is relevant to $\boldsymbol{\omega}_E$, if $\boldsymbol{\omega}_E$ is taken as an optimization variable, $\boldsymbol{\omega}_I$ is also an optimization variable, so the instantaneous CSI NLoS components ${\bf{H}}_{AI}^{NLoS}, {\bf{h}}_{IB}^{NLoS}, {\bf{h}}_{IE}^{NLoS}$ are unavailable for Alice, which should be taken as random variables. In this case the CSI model is expressed as (3). However, once $\boldsymbol{\omega}_I$ and $\boldsymbol{\omega}_E$ are fixed and
for optimizing ${\bf f}$ and $\boldsymbol{\phi}$, only ${\bf{h}}_{IE}^{NLoS}$ are unavailable and the CSI model now is expressed as (4). Therefore, the problem is similar to the previous problem (\ref{ PerfectLocation}) and the AO method could not be applied directly to solve this problem.
 Hence, we still utilize the two-stage method to solve problem (\ref{worst_Eve}), where it jointly optimizes $\boldsymbol{\omega}_I$ and $\boldsymbol{\omega}_E$ to obtain $\boldsymbol{\omega}_I$ and worst $\boldsymbol{\omega}_E$ in the first stage rather than optimize them in two independently subproblems, and then to optimize ${\bf f}$ and $\boldsymbol{\phi}$ to maximize $R$ by fixed IRS and Eve at these locations. Therefore, the objective function could be rewritten as:
\begin{align}
\max \limits_{{\bf{f}},\boldsymbol{\phi}} \left( \max_{\boldsymbol{\omega}_I} \min \limits_{\boldsymbol{\omega}_E}  R \right),
\end{align}
where in the first stage the sub-optimization problem in the bracket is solved, where both the worst location of Eve and IRS deployment location are obtained, and in the second stage the out bracket subproblem is solved, where the optimal beamformer and IRS phase shifts are obtained.

With the same analysis in Sec. \ref{Two-Stage}, in the first stage, we still transform the outage constraints (\ref{worst_Eveb}) with its upper bound based on proposition 3 and 4 to make it only related to $\left(\boldsymbol{\omega}_I,\boldsymbol{\omega}_E\right)$ but not to ${\bf{f}}$ and $\boldsymbol{\phi}$. We denote $\gamma_E \triangleq\frac{\kappa^2 M^2 N_tP}{\left(\kappa+1\right)^2}+|\boldsymbol{\phi}^H\widetilde{\bf{G}}_{AE}^{small}{\bf{f}}|^2+\frac{2\kappa M \sqrt{N_tP}}{\kappa+1}|\boldsymbol{\phi}^H\widetilde{\bf{G}}_{AE}^{small}{\bf{f}}|$ and $\gamma_B\triangleq \frac{\kappa^2 M^2 N_tP}{\left(\kappa+1\right)^2}+|\boldsymbol{\phi}^H\widetilde{\bf{G}}_{AB}^{small}{\bf{f}}|^2-\frac{2\kappa M \sqrt{N_tP}}{\kappa+1}|\boldsymbol{\phi}^H\widetilde{\bf{G}}_{AB}^{small}{\bf{f}}|$ in proposition 3 and 4, and denote cumulative distribution functions of $\gamma_E$ and $\gamma_B$ as as $F_E(x)$ and $F_B(x)$, which are determinate. Then, by denoting $\alpha_E \triangleq F^{-1}_E(1-p_{out})$ and $\alpha_B \triangleq F^{-1}_B(p_{out})$, where $F^{-1}$ is the inverse cumulative distribution function, (\ref{worst_Eveb}) is transformed as $\alpha_E\leq \frac{\left(2^{R_E}-1\right)\sigma^2}{L_{AI}L_{IE}}$ and $\alpha_B\geq \frac{\left(2^{R_B}-1\right)\sigma^2}{L_{AI}L_{IB}}$, so $\left(\boldsymbol{\omega}_I,\boldsymbol{\omega}_E\right)$ only impact $L_{AI}, L_{IB}, L_{IE}$. Then, based on proposition 5, the problem for optimizing $\left(\boldsymbol{\omega}_I,\boldsymbol{\omega}_E\right)$ can be transformed as
\begin{align}
\label{NL_worst}
 \max \limits_{\boldsymbol{\omega}_I} \min \limits_{\boldsymbol{\omega}_E} \ \frac{\sigma^2+\alpha_BL_{AI}L_{IB}}{\sigma^2+\alpha_EL_{AI}L_{IE}}, \ {\rm{s.t.}} \ \boldsymbol{\omega}_I\in {\Omega}_I,\ \boldsymbol{\omega}_E\in \Omega_E.
\end{align}
where $L_{AI}, L_{IB}$ are functions of $\boldsymbol{\omega}_I$, and $L_{IE}$ is a function of both $\boldsymbol{\omega}_I$ and $\boldsymbol{\omega}_E$.
Since $L_{AI}, L_{IB}$ are only related to $\boldsymbol{\omega}_I$, but $L_{IE}$ in numerator is related to both $\boldsymbol{\omega}_I$ and $\boldsymbol{\omega}_E$, $\boldsymbol{\omega}_I$ is interacted with $\boldsymbol{\omega}_E$, so this Max-Min problem is hard to be solved directly. Since this problem is a maximum problem for $\boldsymbol{\omega}_I$, but a minimum problem for $\boldsymbol{\omega}_E$,  the monotonicity cannot be guaranteed by exploited the AO algorithm. To solve it, we transform the original Max-Min problem into several parallel Min-problems by given different $\boldsymbol{\omega}_I$ and then global search for the maximum value over $\boldsymbol{\omega}_I$.

The Min-problem for optimized $\boldsymbol{\omega}_E$ by given any $\boldsymbol{\omega}_I$ is expressed as
\begin{align}
\label{Eve_location}
 \min \limits_{\boldsymbol{\omega}_E} \ \frac{\sigma^2+\alpha_BL_{AI}L_{IB}}{\sigma^2+\alpha_EL_{AI}L_{IE}}, \ {\rm{s.t.}} \ \boldsymbol{\omega}_E\in \Omega_E.
\end{align}
where $\boldsymbol{\omega}_E$ only impacts $L_{IE}$ but does not impact $L_{AI}, L_{IB}$. Hence, we only optimize $\boldsymbol{\omega}_E$ to maximize $L_{IE}$ and this problem is convex which can be solved directly.

 After obtaining $\boldsymbol{\omega}_I$ and $\boldsymbol{\omega}_E$, in the second stage, we optimize $\left({\bf{f}},\boldsymbol{\phi}\right)$ to maximize SR subject to all the constraints with the same method as in Sec. \ref{BeamPhase}.

\section{Simulation Results}
To validate the performance of proposed two stage algorithms, extensive simulation results have been carried out in this section. The system parameters are listed in Table II. Note that all the simulation results illustrated except Fig.2 are averaged over 100 randomly generated channels.
\begin{table}[h]
		\centering
      \caption{simulation parameter}
		\label{tab:Margin_settings}
		\begin{tabular}{|l|c|}\hline
			Carrier center frequency & 2.4GHz\\ \hline
			Path loss exponents & $\rho_{AI}=2.2$, $\rho_{IE}=\rho_{IB}=3$\\ \hline
			Noise power at Bob and Eve & $\sigma_b^2 =\sigma_e^2= -95$ dBm\\ \hline
			Outage probability  & $p_{out}=0.05$ \\ \hline
            Number of transmit antenna & $N_t=4$  \\ \hline
             \multirow{2}{*}{ Location of Bob and Eve} & $\boldsymbol{\omega}_B=(100,15)$, \\
               &  $\boldsymbol{\omega}_E=(95,13)$ \\ \hline
            \multirow{2}{*}{ Location area of IRS} & $x_I\in \Delta_{x_I}=\left[0, 105\right]$, \\
               &  $ y_I \in \Delta_{y_I}=\left[20,30\right]$ \\ \hline
		\end{tabular}
\end{table}

We demonstrate the advantage of the proposed algorithm by comparing its performance with the following three benchmark schemes: 1) Random Location: $\boldsymbol{\omega}_I$ is randomly selected and ${\bf{f}}$ and $\boldsymbol{\phi}$ are optimized by solving $(\ref{beam})$ and $(\ref{shift})$; 2) Global Search Method: $\boldsymbol{\omega}_I$ is optimized by global search solving problem (\ref{max_location}); 3) MRT Method: $\boldsymbol{\omega}_I$ is optimized by proposed method, but ${\bf{f}}$ and $\boldsymbol{\phi}$ are only designed as ${\bf{f}}= \sqrt{\frac{1}{N_t}}{\boldsymbol{\alpha}}_A(\varphi_{AI})$ and $\boldsymbol{\phi}={\rm diag}\left( \boldsymbol{\alpha}_I(\varphi_{IB})\right)\boldsymbol{\alpha}_{I}\left(\theta_{AI}\right)$; 4) Gaussian Random method: We use global search method to optimize the location of IRS, then use the SDR+Guassian Random method to solve the two subproblems for optimizing ${\bf f}$ and $\boldsymbol{\phi}$.

\begin{figure}
\centering
\begin{minipage}[t]{0.48\textwidth}
\centering
\includegraphics[width=3.2in]{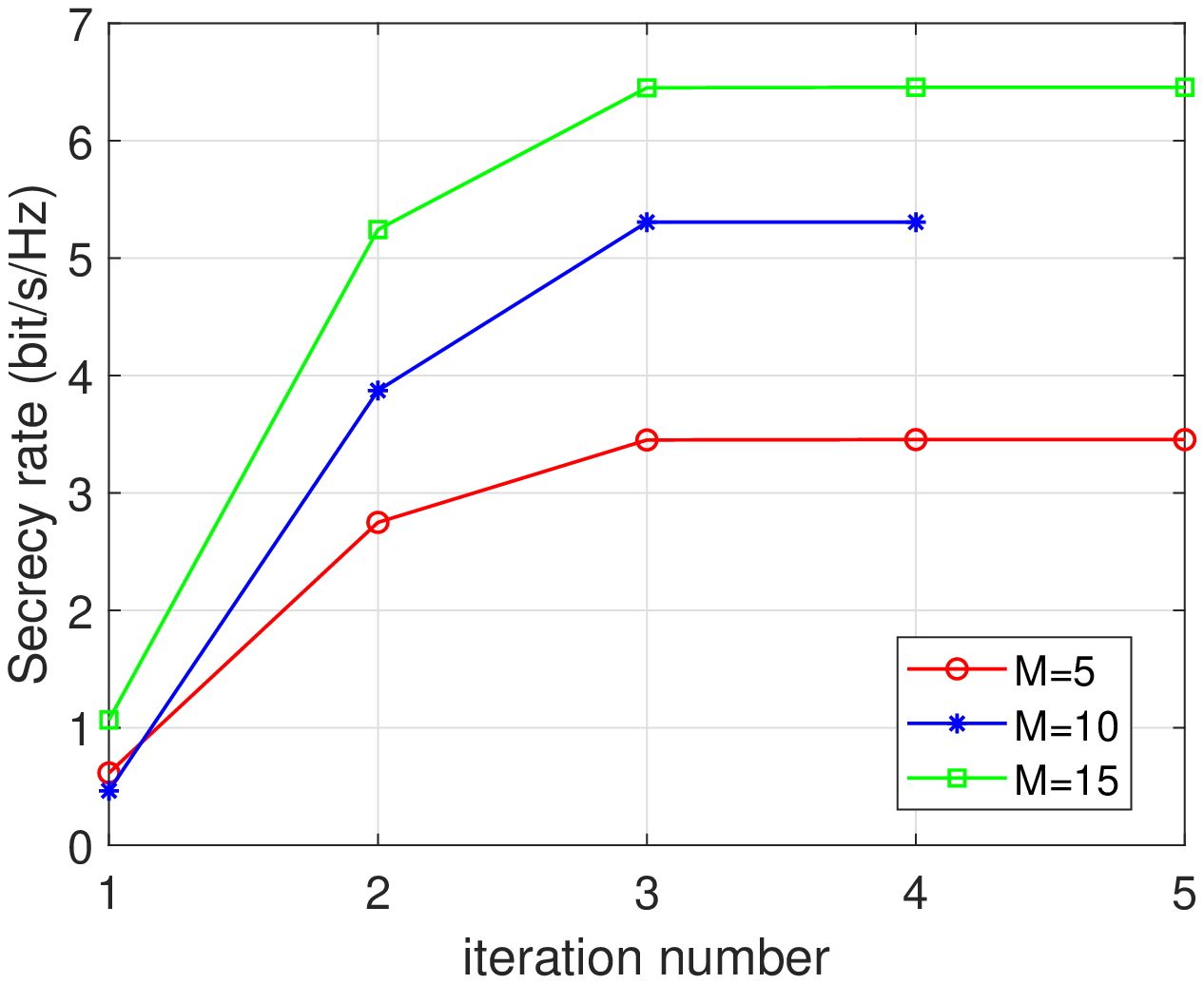}
\caption{Convergence of the proposed two stage scheme with different values of $M$. We set $P=30$dBm, $\kappa=2$.}
\end{minipage}
\begin{minipage}[t]{0.48\textwidth}
\centering
\includegraphics[width=3.2in]{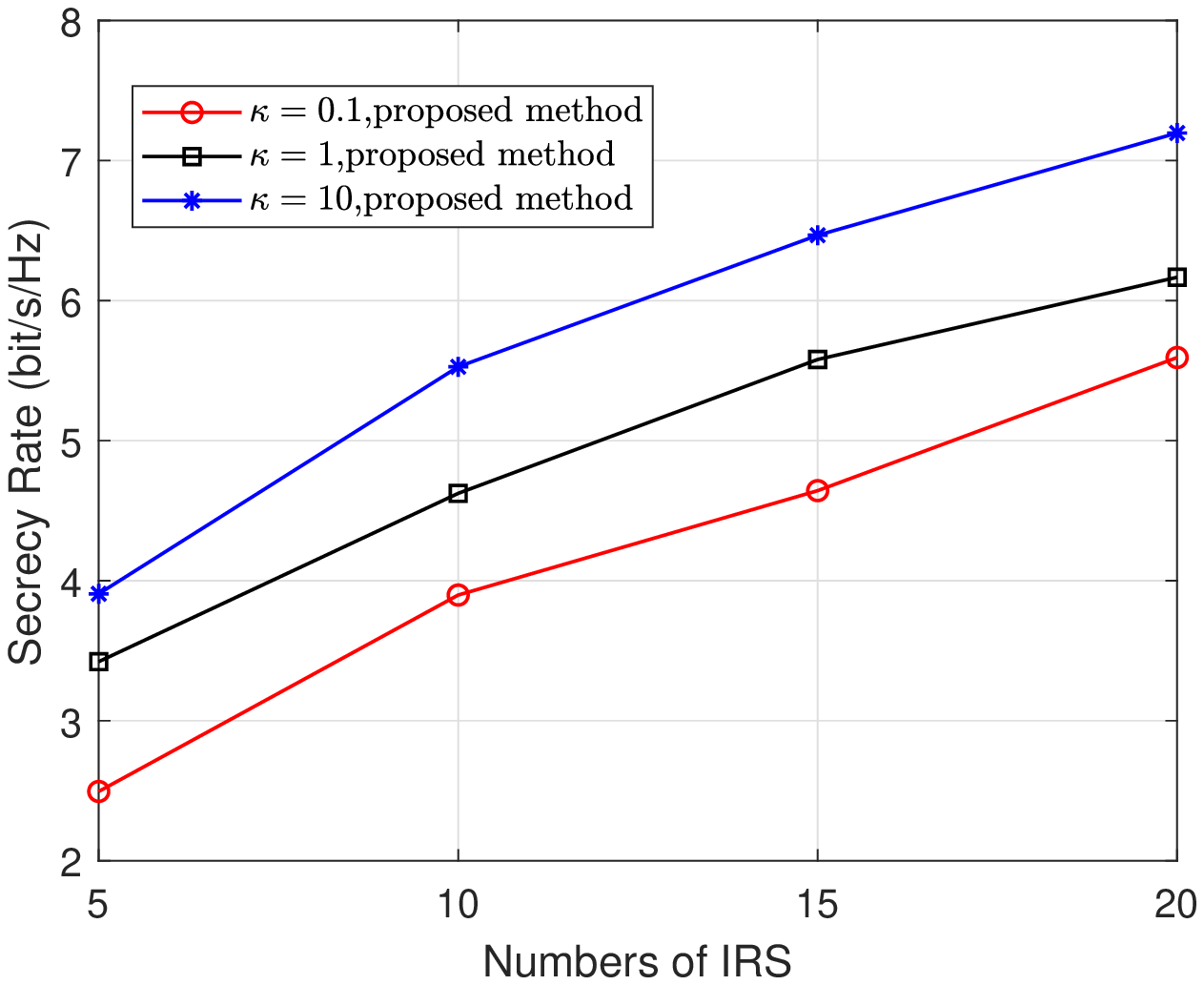}
\caption{Secrecy rate versus number of IRS with different values of $\kappa$. We set $P=30$dBm.}
\end{minipage}
\end{figure}

\begin{figure*}
  \centering
  \subfigure[ With the location of Eve $\boldsymbol{\omega}_{E}=(95,13)$. ]{\includegraphics[width=3.2in]{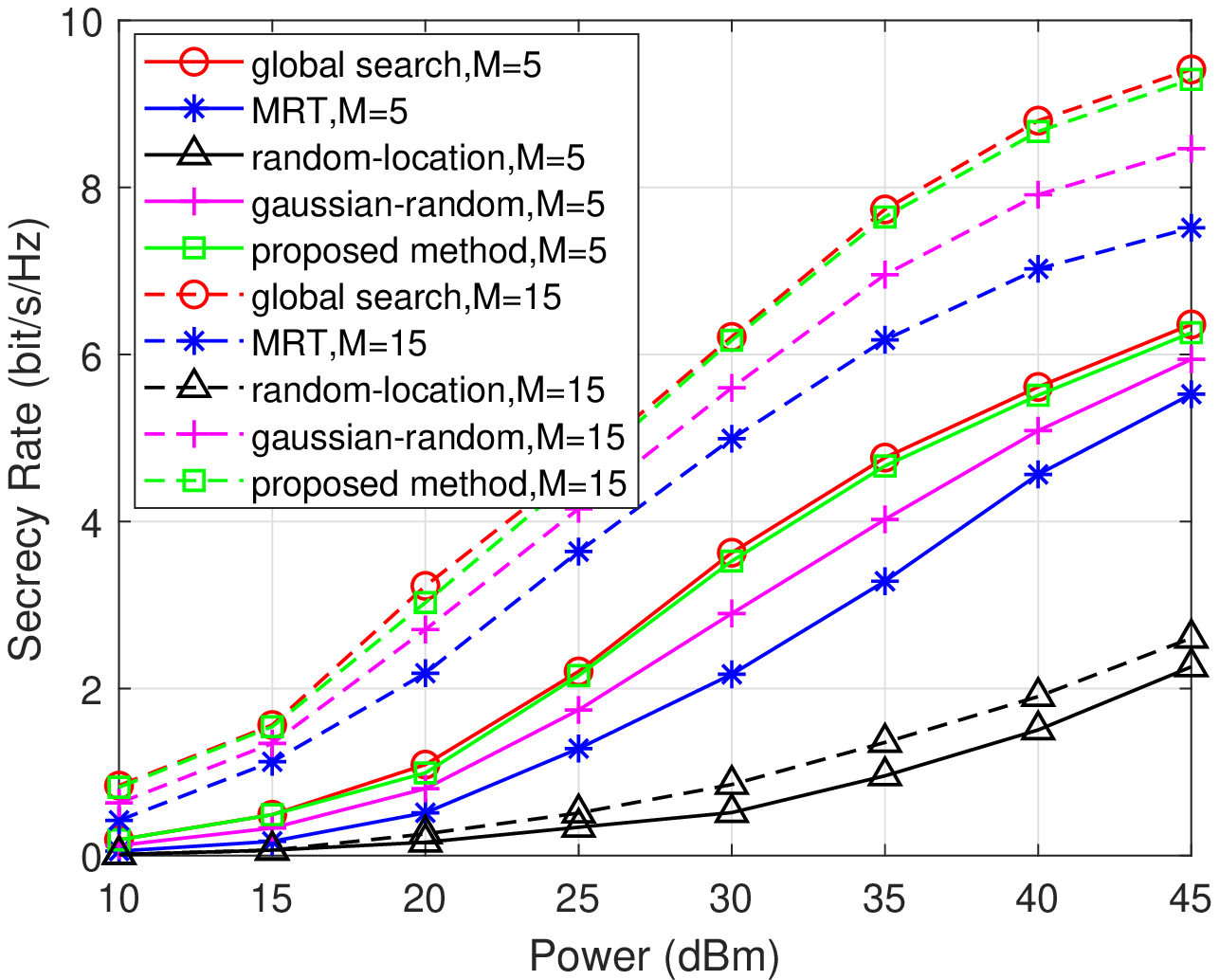}}
  \subfigure[ With the suspicious area of Eve $x_E \in \left(50, 98\right), y_E\in\left(5,13\right)$.]{\includegraphics[width=3.2in]{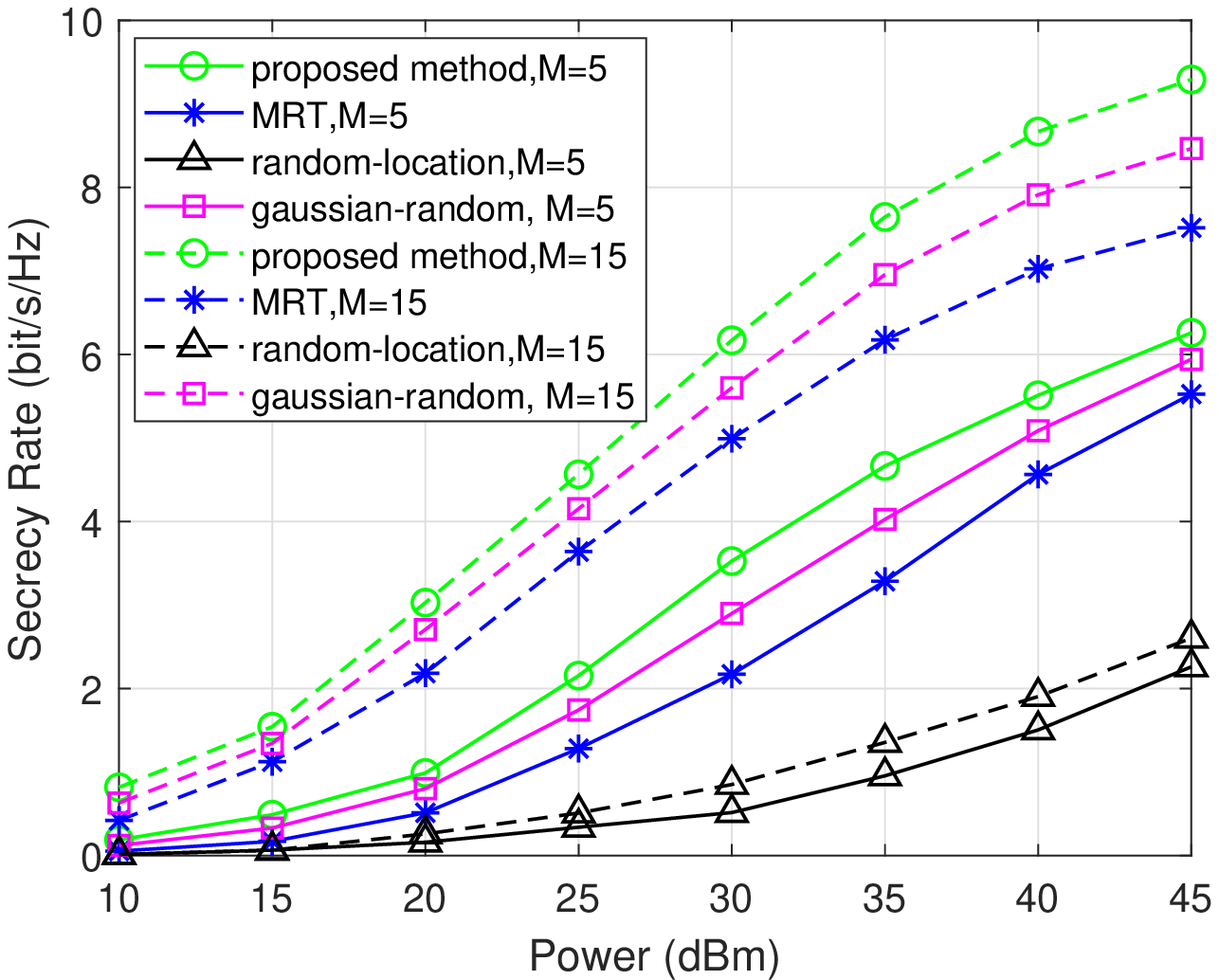}}
  \caption{Secrecy rate versus the transmit power with different values of $M$. We set $\kappa=2$.}
\end{figure*}

\begin{figure*}
  \centering
  \subfigure[ $\boldsymbol{\omega}_I$ versus $\boldsymbol{\omega}_E$ via fixed $\boldsymbol{\omega}_B=(100,15)$ m.]{\includegraphics[width=3.2in]{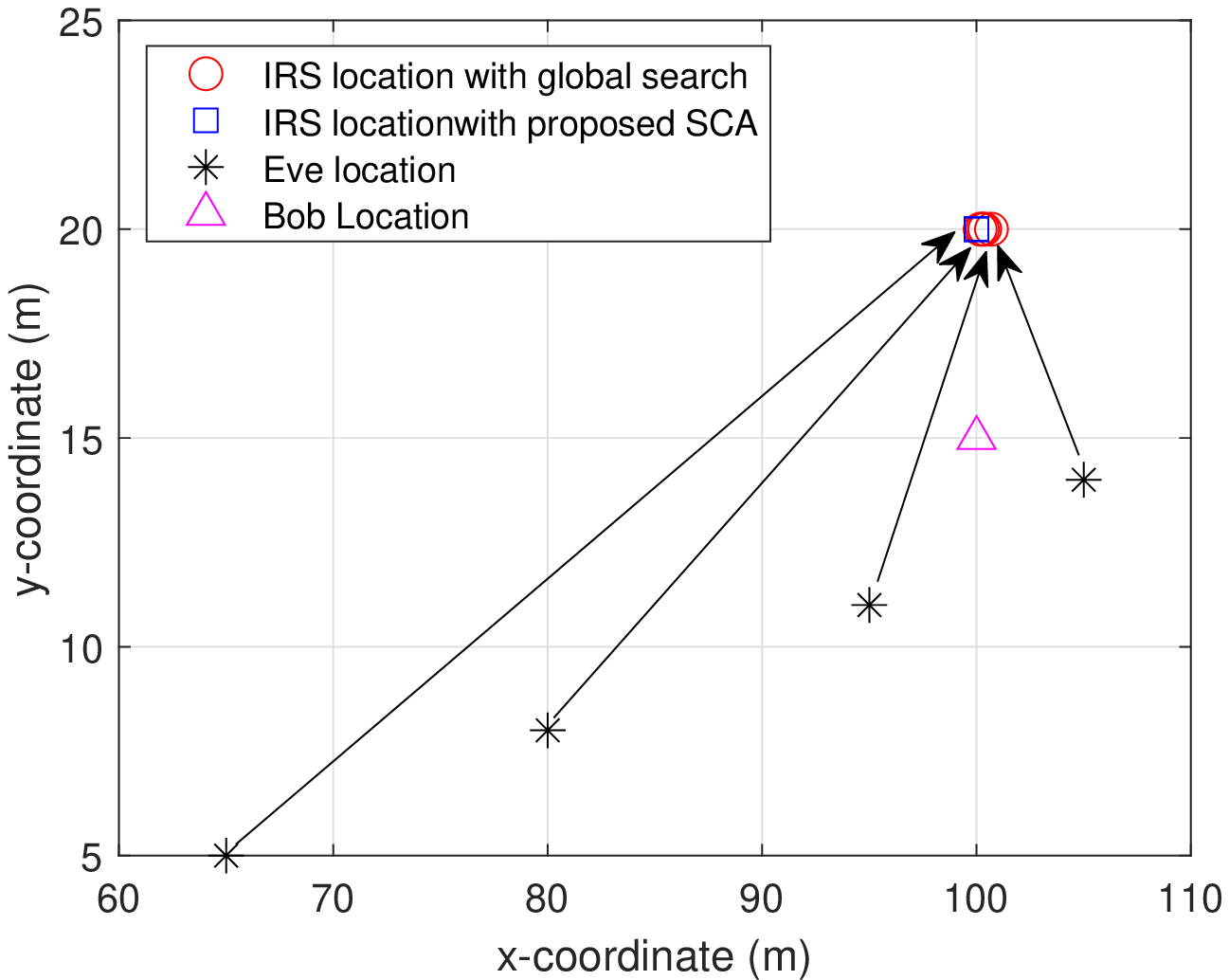}}
  \subfigure[ $\boldsymbol{\omega}_I$ versus $\boldsymbol{\omega}_B$ via fixed $\boldsymbol{\omega}_{E}=(95,13)$ m.]{\includegraphics[width=3.2in]{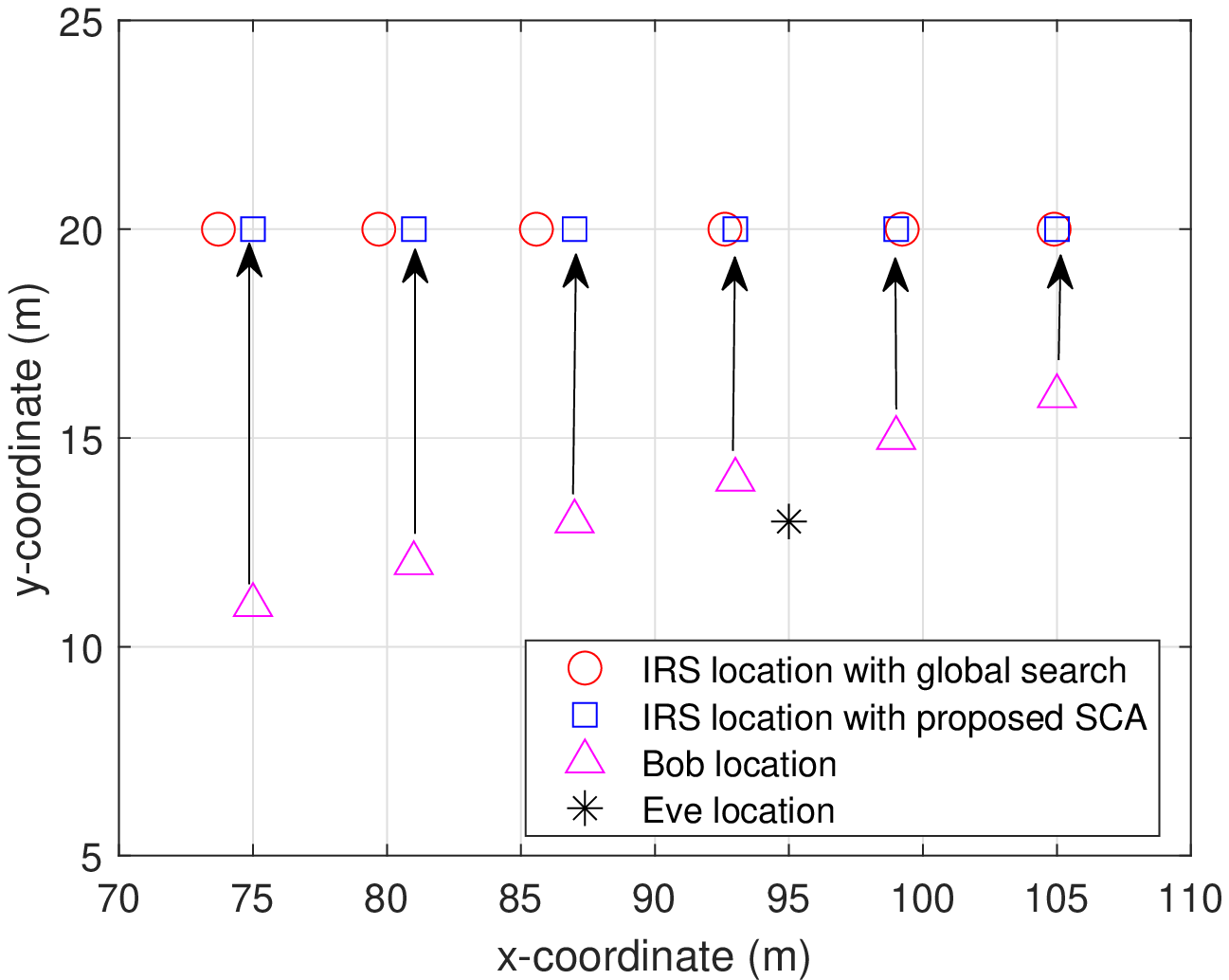}}
  \caption{The location of IRS versus location of Eve and Bob with two schemes. We set $P=30$dBm, $\kappa=2$.}
\end{figure*}

\begin{figure*}
  \centering
  \subfigure[ The location of IRS and Eve with suspicious areas of Eve.]{\includegraphics[width=3.2in]{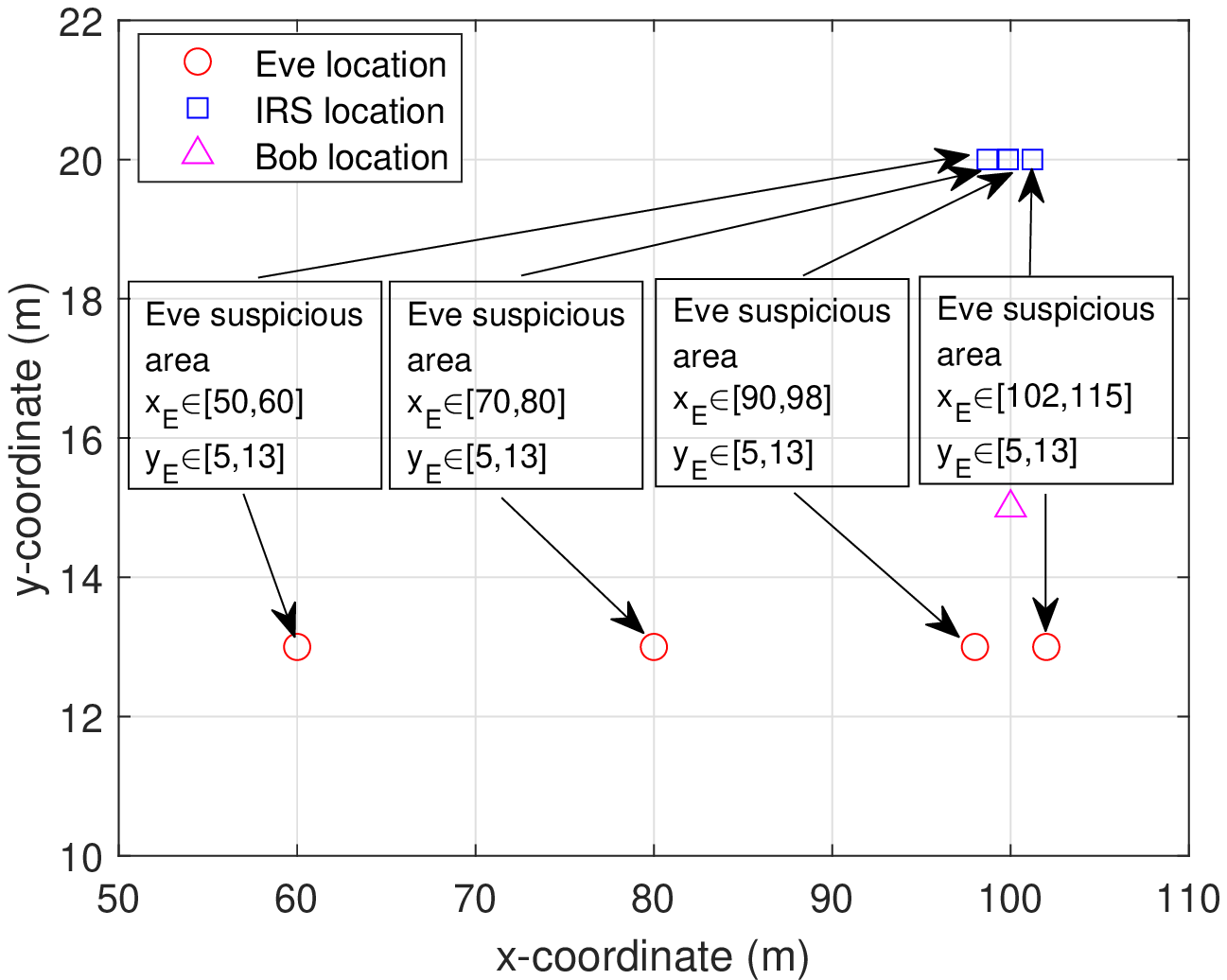}}
  \subfigure[ The secrecy rate with different suspicious areas of Eve.]{\includegraphics[width=3.2in]{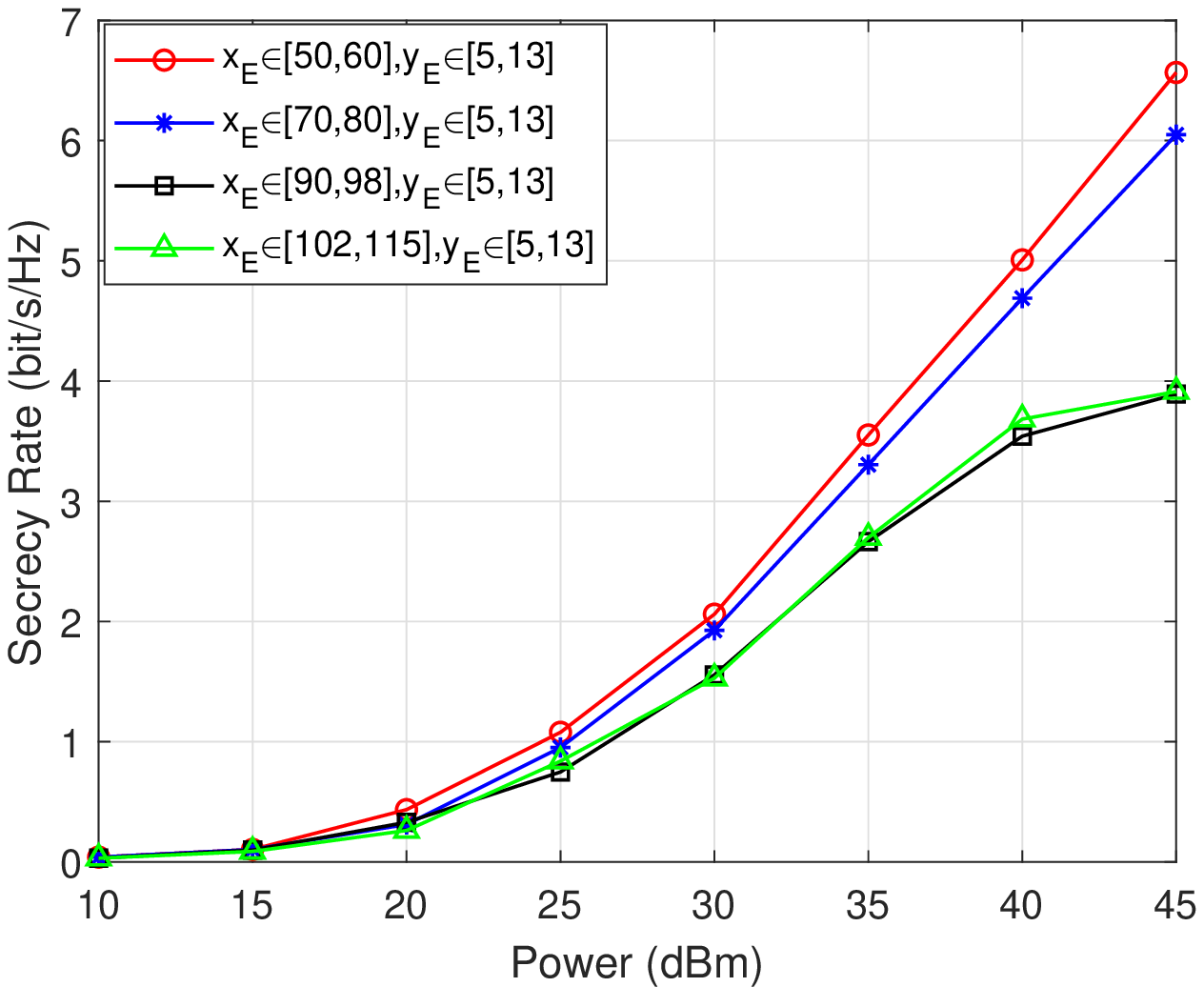}}
  \caption{The location of IRS and Eve and secrecy rate with different suspicious areas of Eve via fixed $\boldsymbol{\omega}_B=(100,15)$ m.}
\end{figure*}

The convergence of the proposed two stage method with different numbers of IRS is investigated in Fig.2. It is seen that the proposed algorithm converges for different element of IRS. With increasing $M$, the dimensions of optimization variables $\boldsymbol{\phi}$ increase, resulting in the computation time increasing.

Fig.3 shows the SR versus number of IRS with different values of $\kappa$. With increased IRS elements, more additional reflecting power can be applied to
transmit signal, thus increasing SR. In addition, we find that SR decreases with decreased $\kappa$. This is because the randomness of NLoS component of Eve link increases with decreased $\kappa$, thus increasing information leakage to Eve.

In Fig.4, we show the maximum SR versus the transmit power with different number of IRS and with different schemes in two cases. As observed from Fig.4(a) and Fig.4(b), we find that the maximum SR increases monotonically with the transmit power and elements of IRS. This is because larger SR is required more transmit power and more additional reflecting power can be applied to transmit signal with increased IRS elements. Compared to the random location scheme, our scheme can significantly enhance SR. This is due to the fact that if the location of IRS is randomly chosen, the path loss of both Eve and Bob may be large and LoS components of Eve and Bob may be more similar, so it is hard to guarantee security. Therefore, an optimized IRS location not only reduce the path loss, but enhance the superiority of legitimate channels. In addition, from Fig.4(a), we find that the proposed scheme with low complexity and the global search method with high complexity have the same performance.

Fig.5 depicts the location of IRS $\boldsymbol{\omega}_I$ versus the location of Eve $\boldsymbol{\omega}_E$ and Bob $\boldsymbol{\omega}_B$. In Fig.5(a) and Fig.5(b), we show the $\boldsymbol{\omega}_I$ optimized by our proposed method  and use global search method as benchmarks under four different $\boldsymbol{\omega}_E$ by fixed $\boldsymbol{\omega}_B=(100,15)$ and under six different $\boldsymbol{\omega}_B$ by fixed $\boldsymbol{\omega}_E=(95,13)$. From these two figures, we find that the location of IRS optimized by proposed  method is very close to that optimized by global search method. Moreover, the figures show that IRS should be deployed as close to Bob as possible wherever Eve is located. This is because that when IRS is deployed closely to Bob, the large-scale path loss of Alice-IRS-Bob link is small and the quality of legitimate channel is enhanced significantly.

Fig.6 shows the locations of IRS and Eve and secrecy rate versus suspicious areas of Eve. Fig.6(a) shows the $\boldsymbol{\omega}_I$ and $\boldsymbol{\omega}_E$ optimized by SCA and global search method under four different suspicious areas of Eve by fixed $\boldsymbol{\omega}_B=(100,15)$ m. We find that the worst $\boldsymbol{\omega}_E$ is the location closest to IRS in each the suspicious area of Eve. This is due to the fact that when Eve is closest to IRS, the path loss of IRS-Eve is minimum and information leakage to Eve is maximum. In Fig.6(b), by fixed $M=5$, we find that when suspicious area of Eve is closer to Bob, the SR is smaller. This is because that when  the suspicious area of Eve is closer to Bob, the worst $\boldsymbol{\omega}_E$ is closer to Bob and IRS, thus increasing the information leakage to Eve and decreasing the SR.

\section{Conclusion}
In this paper, we investigated the robust secrecy transmission in the IRS-aided multiple antennas wireless communications. For the first time the location optimization was considered in this work and we aim to maximize the secrecy rate by optimizing the location of IRS, transmit beamformer at Alice and phase shifts at IRS under two
different cases with the location of Eve or not.  We show the joint optimization problem could be  solved thought a two-stage optimization framework. In the first stage,  IRS location could be optimized via exploiting successive convex approximation method. In the second stage, an AO algorithm is proposed to optimize beamformer and phase shifts iteratively.
Similar idea has also been developed to solve the case where only a suspicious area of Eve is known.  Simulation results have verified the effectively of the proposed algorithm and shown the importance of IRS location optimization for enhancing secrecy performance.

\section{Appendix}
\subsection{ Proof of Proposition 1}

 First, for ${\rm diag}\left(\overline{ \bf{h}}_{IJ} \right)\widetilde{\bf{H}}_{AI}$ in (\ref{AllError}), each element of $\widetilde{\bf{H}}_{AI}$ follows $\mathcal{CN}\left({ 0},\frac{L_{AI}}{\kappa+1}\right)$ and each element of $\overline{\bf{h}}_{IJ}$ is determinate whose modulus is $\frac{\kappa L_{IJ}}{\kappa+1}$, so each element of ${\rm diag}\left(\overline{ \bf{h}}_{IJ} \right)\widetilde{\bf{H}}_{AI}$ follows $ \mathcal{CN}\left({ 0},\frac{\kappa L_{IJ}L_{AI}}{\left(\kappa+1\right)^2}\right)$. Then, similarly to ${\rm diag}\left(\widetilde{\bf{h}}_{IJ}\right)\overline{\bf{H}}_{AI}$, $\widetilde{\bf{h}}_{IJ}\sim \mathcal{CN}\left({\bf 0},\frac{L_{IJ}}{\kappa+1}{\bf I}_M\right)$ and each element of $\overline{\bf{H}}_{AI}$ is determinate with constant modulus $\frac{\kappa L_{AI}}{\kappa+1}$, so each element of ${\rm diag}\left(\widetilde{\bf{h}}_{IJ}\right)\overline{\bf{H}}_{AI}$ approximately follows $ \mathcal{CN}\left({0},\frac{\kappa L_{IJ}L_{AI}}{\left(\kappa+1\right)^2}{}\right)$. Hence, each element of $\widetilde{\bf{G}}_{AJ}$ follows $\mathcal{CN}\left({0},\frac{2\kappa L_{AI}L_{IJ}}{(\kappa+1)^2}\right)$.
\subsection{ Proof of Proposition 2}
First, according to (\ref{rate}), we re-write $\boldsymbol{\phi}^H \widetilde{\bf{G}}_{AJ} {\bf{f}}=\widetilde{\bf{h}}_{IJ}^H \boldsymbol{\Phi}\overline{\bf{H}}_{AI}{\bf{f}} + \overline{\bf{h}}_{IJ}^H \boldsymbol{\Phi}\widetilde{\bf{H}}_{AI}{\bf{f}}$. Based on proposition 1, we have $\widetilde{\bf{h}}_{IJ}\sim \mathcal{CN}\left({\bf 0},\frac{L_{IJ}}{\kappa+1}{\bf I}_{M}\right)$ and the each element of $\widetilde{\bf{H}}_{AI}$ follows $\mathcal{CN}\left({ 0},\frac{L_{AI}}{\kappa+1}\right)$. For any unitary matrices $\boldsymbol{\Phi}$, we have $
\widetilde{\bf h}_{IJ}^H\boldsymbol{\Phi} \sim \mathcal{CN}\left({\bf 0},\frac{L_{IJ}}{\kappa+1}{{\bf I}_{M}}\right)$ and each element of $\ \boldsymbol{\Phi}\widetilde{\bf H}_{AI}$ also follows $ \mathcal{CN}\left({0},\frac{L_{AI}}{\kappa+1}{}\right),$
  so that $\widetilde{\bf h}_{IJ}^H\boldsymbol{\Phi}$ and $\boldsymbol{\Phi}\widetilde{\bf H}_{AI}$ have the same distributions with $\widetilde{\bf h}_{IJ}^H$ and $\widetilde{\bf H}_{AI}$. Then, we have $\widetilde{\bf{h}}_{IJ}^H \boldsymbol{\Phi}\overline{\bf{H}}_{AI} \sim \overline{\bf{h}}_{IJ}^H \boldsymbol{\Phi}\widetilde{\bf{H}}_{AI} \sim \mathcal{CN}\left({\bf 0},\frac{\kappa L_{IJ}L_{AI}M}{\left(\kappa+1\right)^2}{\bf I}_{N_t}\right)$. For brief, we denote ${\bf x}=\widetilde{\bf{h}}_{IJ}^H \boldsymbol{\Phi}\overline{\bf{H}}_{AI}$. For any ${\bf f}$,
$\widetilde{\bf{h}}_{IJ}^H \boldsymbol{\Phi}\overline{\bf{H}}_{AI}{\bf{f}}=\sum \limits_{i=1}^{N_t} x_if_i$, where $x_i\sim \mathcal{CN}\left({ 0},\frac{\kappa L_{IJ}L_{AI}M}{\left(\kappa+1\right)^2}\right)$, so $\sum \limits_{i=1}^{N_t} x_if_i \sim \mathcal{CN}\left({0},\sum\limits_{i=1}^{N_t}\frac{\kappa L_{IJ}L_{AI}M|f_i|^2}{\left(\kappa+1\right)^2}\right)$. Since $\| {\bf f}\|^2=P$, we have $\widetilde{\bf{h}}_{IJ}^H \boldsymbol{\Phi}\overline{\bf{H}}_{AI}{\bf{f}}\sim \mathcal{CN}\left({ 0},\frac{\kappa L_{IJ}L_{AI}MP}{\left(\kappa+1\right)^2}\right)$. Similarly, $\overline{\bf{h}}_{IJ}^H \boldsymbol{\Phi}\widetilde{\bf{H}}_{AI}{\bf{f}}$ follows $\mathcal{CN}\left({ 0},\frac{\kappa L_{IJ}L_{AI}MP}{\left(\kappa+1\right)^2}\right)$. Hence, for any $\boldsymbol{\phi}$ and ${\bf f}$, $\boldsymbol{\phi}^H \widetilde{\bf{G}}_{AJ}{\bf{f}}\sim \mathcal{CN}\left({ 0},\frac{2\kappa L_{IJ}L_{AI}MP}{\left(\kappa+1\right)^2}\right)$.
\subsection{ Proof of Proposition 3}
For the left constraint of (\ref{out}), according to the triangle inequality, we find
\begin{align}
\notag
&\big |\boldsymbol{\phi}^H \left(\overline{\bf{G}}_{AE}+\widetilde{\bf{G}}_{AE} \right){\bf{f}}\big |^2 \\
\notag
&\leq |\boldsymbol{\phi}^H\overline{\bf{G}}_{AE} {\bf{f}}|^2+|\boldsymbol{\phi}^H\widetilde{\bf{G}}_{AE}{\bf{f}}|^2+2|\boldsymbol{\phi}^H\overline{\bf{G}}_{AE}  {\bf{f}}||\boldsymbol{\phi}^H\widetilde{\bf{G}}_{AE}{\bf{f}}|\\
\notag
&\overset{(d)}\leq  L_{AI}L_{IE}\left(\frac{\kappa M\sqrt{N_tP}}{\kappa+1}+|\boldsymbol{\phi}^H\widetilde{\bf{G}}_{AE}^{small}{\bf{f}}|\right)^2,
\end{align}
where $\widetilde{\bf{G}}_{AE}^{small}\triangleq\frac{\sqrt{\kappa}}{\kappa+1}\left({\rm diag}\left({\bf{h}}_{IE}^{NLoS}\right){\bf{H}}_{AI}^{LoS} + {\rm diag}\left({\bf{h}}_{IE}^{LoS}\right){\bf{H}}_{AI}^{NLoS}\right)\sim \mathcal{CN}\left({ 0},\frac{2\kappa M}{\left(\kappa+1\right)^2}\right)$ based on proposition 1 and any choice of $\boldsymbol{\phi}$ and ${\bf{f}}$ does not impact the distribution of $\boldsymbol{\phi}^H\widetilde{\bf{G}}_{AE}^{small}{\bf{f}}\sim \mathcal{CN}\left({ 0},\frac{2\kappa MP}{\left(\kappa+1\right)^2}\right)$ based on proposition 2. $(d)$ is due to $|\boldsymbol{\phi}^H\overline{\bf{G}}_{AE} {\bf{f}}|^2\leq \|\boldsymbol{\phi}^H\|^2\|\overline{\bf{G}}_{AE}\|^2_F \|{\bf{f}}\|^2= \frac{\kappa^2 L_{AI}L_{IE} M^2 N_tP}{\left(\kappa+1\right)^2}$, where $\|\boldsymbol{\phi}^H\|^2=M$, $\|{\bf{f}}\|^2=P$ and each element of $\overline{ \bf{G}}_{AE}$ has constant modulus $\frac{\kappa \sqrt{L_{AI}L_{IE}}}{\kappa+1}$ analyzed in (\ref{AllError}) so $\|\overline{\bf{G}}_{AE}\|^2_F=\frac{\kappa^2 L_{AI}L_{IE} M N_t}{\left(\kappa+1\right)^2}$. When ${\bf{f}}$ and $\boldsymbol{\phi}$ are adopted to the LoS components in Alice-IRS-Eve channel link $\overline{\bf{G}}_{AE}$, the equation $|\boldsymbol{\phi}^H\overline{\bf{G}}_{AE} {\bf{f}}|^2= \frac{\kappa^2 L_{AI}L_{IE} M^2 N_tP}{\left(\kappa+1\right)^2}$ in $(b)$ holds. It implies that the information leakage for Eve is maximum.

Then, we have the upper bound  as
\begin{align}
\notag
&{\rm{Pr}}\left\{C_E(\boldsymbol{\omega}_I,{\bf{f}},\boldsymbol{\phi})\geq R_E \right\}\\
\notag
 &\leq {\rm{Pr}}\left\{ \left(\frac{\kappa M \sqrt{N_tP}}{\kappa+1}+|\boldsymbol{\phi}^H\widetilde{\bf{G}}_{AE}^{small}{\bf{f}}|\right)^2\geq \frac{\left(2^{R_E}-1\right)\sigma^2}{L_{AI}L_{IE}} \right\}.
 \end{align}
This constraint is irrelevant to $\boldsymbol{\phi}$ and ${\bf{f}}$.
\subsection{ Proof of Proposition 4}
For the right constraint of (\ref{out}), with the triangle inequality, we obtain
\begin{align}
\notag
&\big |{\boldsymbol{\phi}}^H \left(\overline{\bf{G}}_{AB}+\widetilde{\bf{G}}_{AB} \right){\bf{f}}\big|^2 \\
\notag
&\geq |{\boldsymbol{\phi}}^H\overline{\bf{G}}_{AB}
{\bf{f}}|^2+|{\boldsymbol{\phi}}^H\widetilde{\bf{G}}_{AB}{\bf{f}}|^2-2|{\boldsymbol{\phi}}^H\overline{\bf{G}}_{AB}  {\bf{f}}||{\boldsymbol{\phi}}^H\widetilde{\bf{G}}_{AB}{\bf{f}}|,
\end{align}
 Then, we transform the constraint ${\rm{Pr}}\left\{C_B(\boldsymbol{\omega}_I,{\bf{f}},\boldsymbol{\phi}) \leq R_B\right\}$ with its upper bound as
  \begin{align}
\notag
 &{\rm{Pr}}\left\{C_B(\boldsymbol{\omega}_I,{\bf{f}},\boldsymbol{\phi})\leq R_B \right\}\\
 \notag
&\leq {\rm{Pr}}\left\{\log\left(1+\frac{\left(|{\boldsymbol{\phi}}^H\overline{\bf{G}}_{AB} {\bf{f}}|-|{\boldsymbol{\phi}}^H\widetilde{\bf{G}}_{AB}{\bf{f}}|\right)^2 }{{\sigma}^2}\right)
 \leq R_B \right\}.
 \end{align}
  Although the choice of ${\bf{f}}$ and $\boldsymbol{\phi}$ will not change the distribution of $\boldsymbol{\phi}^H \widetilde{\bf{G}}_{AB}{\bf{f}}$ based on proposition 2, but ${\bf{f}}$ and $\boldsymbol{\phi}$ impact the value of $\boldsymbol{\phi}^H \overline{\bf{G}}_{AB}{\bf{f}}$, so ${\bf{f}}$ and $\boldsymbol{\phi}$ impact this constraint together.

  The main challenge is still how to adjust ${\bf{f}}$ and $\boldsymbol{\phi}$. Fortunately, the existing work \cite{Yan-16} implies that for a Rician fading, if ${\bf{f}}$ and $\boldsymbol{\phi}$ will not impact the communication quality of Eve, they should be adopted to the LoS component in Alice-IRS-Bob link for achieving the lowest outage probability. This is because that the outage probability is decreasing with Rician factor. When ${\bf{f}}$ and $\boldsymbol{\phi}$ are adopted to the LoS component in Alice-IRS-Bob link, the power of the deterministic component is maximum, so the Rician factor is maximum.

  Therefore, if we can prove that ${\boldsymbol{\phi}}^H \left(\overline{\bf{G}}_{AB}+\widetilde{\bf{G}}_{AB} \right){\bf{f}}$ follows Rician fading and any ${\bf{f}}$ and $\boldsymbol{\phi}$ will not impact SNR of Eve, ${\bf{f}}$ and $\boldsymbol{\phi}$ should adopt to the LoS component $\overline{\bf{G}}_{AB}$ to obtain the lowest outage probability.
  In fact, ${\boldsymbol{\phi}}^H \left(\overline{\bf{G}}_{AB}+\widetilde{\bf{G}}_{AB} \right){\bf{f}}$ is a Rician channel based on the analysis of (\ref{rate}) with determinate component $\boldsymbol{\phi}^H\overline{\bf{G}}_{AB}{\bf{f}}$ and random component $\boldsymbol{\phi}^H\widetilde{\bf{G}}_{AB}{\bf{f}}$. And based on proposition 4, the upper bound of ${\rm{Pr}}\left\{C_E(\boldsymbol{\omega}_I,{\bf{f}},\boldsymbol{\phi}) \geq R_E\right\}$ is irrelevant to $\boldsymbol{\phi}$ and ${\bf{f}}$.
  Hence, by adjusting ${\bf{f}}^*= \sqrt{\frac{P}{N_t}}{\boldsymbol{\alpha}}_A(\varphi_{AI})$ and $\boldsymbol{\phi}^*={\rm diag}\left(
  \boldsymbol{\alpha}_I(\varphi_{IB})\right)\boldsymbol{\alpha}_{I}\left(\theta_{AI}\right)$, the lowest outage probability is achieved. By setting the optimal ${\bf{f}}^*$ and ${\boldsymbol{\phi}}^*$, we obtain
  \begin{align}
\notag
 &{\rm{Pr}}\left\{\log\left(1+\frac{\left(|{\boldsymbol{\phi}}^H\overline{\bf{G}}_{AB} {\bf{f}}|-|{\boldsymbol{\phi}}^H\widetilde{\bf{G}}_{AB}{\bf{f}}|\right)^2 }{{\sigma}^2}\right)
 \leq R_B \right\}\\
\label{PB}
&={\rm{Pr}}\left\{ \left( \frac{\kappa M\sqrt{N_tP}}{\kappa+1}-|\boldsymbol{\phi}^{H*}\widetilde{\bf{G}}_{AB}^{small}{\bf{f}^*}| \right)^2\leq \frac{\left(2^{R_B}-1\right)\sigma^2}{L_{AI}L_{IB}} \right\}.
 \end{align}
 where $\widetilde{\bf{G}}_{AB}^{small}\triangleq\frac{\sqrt{\kappa}}{\kappa+1}\left({\rm diag}\left({\bf{h}}_{IB}^{NLoS}\right){\bf{H}}_{AI}^{LoS} + {\rm diag}\left({\bf{h}}_{IB}^{LoS}\right){\bf{H}}_{AI}^{NLoS}\right) \sim \mathcal{CN}\left({ 0},\frac{2\kappa M}{\left(\kappa+1\right)^2}\right)$ based on proposition 1. Since the choice of ${\bf{f}}$ and $\boldsymbol{\phi}$ will not change distribution of $\boldsymbol{\phi}^H \widetilde{\bf{G}}_{AB}^{small}{\bf{f}}$, so ${\boldsymbol{\phi}^H}^* \widetilde{\bf{G}}_{AB}^{small}{\bf{f}}^*$ still follows $\mathcal{CN}\left({ 0},\frac{2\kappa MP}{\left(\kappa+1\right)^2}\right)$ based on proposition 2.
 After transforming ${\rm{Pr}}\left\{C_B(\boldsymbol{\omega}_I,{\bf{f}},\boldsymbol{\phi}) \leq R_B\right\}$ with (\ref{PB}), this constraint is irrelevant to ${\bf{f}}$ and ${\boldsymbol{\phi}}$.

\subsection{Proof of Proposition 5}
 First, we transform the two constraints $\alpha_E\leq \frac{\left(2^{R_E}-1\right)\sigma^2}{L_{AI}L_{IE}}$ and $\alpha_B\geq \frac{\left(2^{R_B}-1\right)\sigma^2}{L_{AI}L_{IB}}$ as $R_E\geq \log\left( 1+\frac{\alpha_EL_{AI}L_{IE}}{\sigma^2}\right)$
and $R_B\leq \log\left( 1+\frac{\alpha_BL_{AI}L_{IB}}{\sigma^2}\right)$. Then, (\ref{AO_method1}) is transformed as
\begin{align}
\notag
&\max \limits_{\boldsymbol{\omega}_I \in \Omega_{I},R_B,R_E} \ R_B-R_E, \\
\notag
&\ s.t. \ R_E\geq \log\left( 1+\frac{\alpha_EL_{AI}L_{IE}}{\sigma^2}\right),\ R_B\leq \log\left( 1+\frac{\alpha_BL_{AI}L_{IB}}{\sigma^2}\right).
 \end{align}
 To maximize $R_B-R_E$ is equivalent to maximize $\log\left( 1+\frac{\alpha_BL_{AI}L_{IB}}{\sigma^2}\right)-\log\left( 1+\frac{\alpha_EL_{AI}L_{IE}}{\sigma^2}\right)$. Hence, the above problem is equivalently transformed as
\begin{align}
&\max \limits_{\boldsymbol{\omega}_I\in {\Omega}_I} \ \frac{\sigma^2+\alpha_BL_{AI}L_{IB}}{\sigma^2+\alpha_EL_{AI}L_{IE}}.
\end{align}
\subsection{ Proof of Proposition 6}
For briefly, we denote ${\bf{h}}_b\triangleq\left({\bf{h}}_{IB}^H\boldsymbol{\Phi}{{\bf{H}}}_{AI}\right)^H\in \mathbb{C}^{M\times 1}$ and ${\bf{h}}_e \triangleq \left({\bf{h}}_{IE}^H\boldsymbol{\Phi}
{{\bf{H}}}_{AI}\right)^H\in \mathbb{C}^{M\times 1}$. Then, we assume that ${\bf{F}}^*$ is the optimal solution of $(\ref{beam2})$ and construct a new solution $\widetilde{\bf{F}}^*\triangleq{\bf{F}}^{*\frac{1}{2}}{\bf{P}}{\bf{F}}^{*\frac{1}{2}}$, where ${\bf{P}}\triangleq\frac{{\bf{F}}^{*\frac{1}{2}}
{\bf{h}}_b{\bf{h}}_b^H{\bf{F}}^{*\frac{1}{2}}}{\left\|{\bf{F}}^{*\frac{1}{2}} {\bf{h}}_b\right\|^2}$ is the projection matrix.
Obviously, $\widetilde{\bf{F}}^*$ is a rank-one matrix. From the value of function ${\bf{F}}^*-\widetilde{\bf{F}}^*={\bf{F}}^{*\frac{1}{2}}\left(
{\bf{I}}-{\bf{P}}\right){\bf{F}}^{*\frac{1}{2}}\succeq {\bf{0}}$, we find ${\rm Tr}\left( \widetilde{\bf{F}}^* \right)\leqslant {\rm Tr}\left({\bf{F}}^* \right)$,
which means that the objective value of $(\ref{beam2})$ obtained by $\widetilde{\bf{F}}^*$ is not worse than that obtained by ${\bf{F}}^*$. Finally, we check whether $\widetilde{\bf{F}}^*$ is satisfying constraint (\ref{BTI}). Since it is computationally
intractable to check whether $\widetilde{\bf{F}}^*$ satisfies the constraint (\ref{BTI}) directly, we instead consider the equivalent constraint (\ref{AOb}), which can be equivalently reformulated as
\begin{align}
{\rm{Pr}}\left\{ \underbrace{\log\left(1+\frac{{\bf{h}}_b^H{\bf{F}}{\bf{h}}_b}{{\sigma}^2}\right)}_{t_1}- \underbrace{\log\left(1+\frac{{\bf{h}}_e^H{\bf{F}}{\bf{h}}_e}{{\sigma}^2}\right)}_{t_2} \geq R \right\} \geq 1- p_{out}.
\end{align}
By substituting $\widetilde{\bf{F}}^*$ into $t_1$, we have
${\bf{h}}_b^H\widetilde{\bf{F}}^*{\bf{h}}_b={\bf{h}}_b^H {\bf{F}}^{*\frac{1}{2}}{\bf{P}}{\bf{F}}^{*\frac{1}{2}} {\bf{h}}_b=\frac{{\bf{h}}_b^H {\bf{F}}^{*\frac{1}{2}}{\bf{F}}^{*\frac{1}{2}}{\bf{h}}_b{\bf{h}}_b^H{\bf{F}}^{*\frac{1}{2}}{\bf{F}}^{*\frac{1}{2}} {\bf{h}}_b}{ {\bf{h}}_b^H {\bf{F}}^{*\frac{1}{2}} {\bf{F}}^{*\frac{1}{2}} {\bf{h}}_b}={\bf{h}}_b^H {\bf{F}}^* {\bf{h}}_b$. Hence, the value of $t_1$ is the same by replacing ${\bf{F}}^*$ with $\widetilde{\bf{F}}^*$. Moreover, we have
${\bf{h}}_e^H{\bf{F}}^*{\bf{h}}_e-{\bf{h}}_e^H\widetilde{\bf{F}}^*{\bf{h}}_e= {\bf{h}}_e^H\left({\bf{F}}^*-{\bf{F}}^{*\frac{1}{2}}{\bf{P}}{\bf{F}}^{*\frac{1}{2}}\right){\bf{h}}_e = {\bf{h}}_e^H\left( {\bf{F}}^{*\frac{1}{2}}\left({\bf{I}}-{\bf{P}}\right){\bf{F}}^{*\frac{1}{2}}\right){\bf{h}}_e\geqslant 0$. Thus, the value of $t_2$ is not increase by substituting $\widetilde{\bf{F}}^*$ into it. Therefore, $\widetilde{\bf{F}}^*$ still satisfies $(\ref{AOb})$ and then satisfies (\ref{BTI}). Hence,
we must obtain a rank-one solution of $(\ref{beam2})$ if the problem is feasible. 


\begin{thebibliography}{99}


\bibitem{Wu-19} Q. Wu and R. Zhang, ``Towards smart and reconfigurable environment: Intelligent reflecting surface aided wireless network," \emph{IEEE Commun. Magazine}, vol. 58, no. 1, pp. 106-112, Jan. 2020.

\bibitem {Wu-19b}  Q. Wu and R. Zhang, ``Intelligent reflecting surface enhanced wireless network via Joint active and passive beamforming," \emph{IEEE Trans. Wireless Commun.},vol. 18, no. 11, pp. 5394-5409, Nov. 2019.

\bibitem {Huang-19} C. Huang, A. Zappone, G. C. Alexandropoulos, M. Debbah, and C. Yuen,  ``Reconfigurable intelligent surfaces for energy efficiency in wireless communication," \emph{IEEE Trans. Wireless Commun.}, vol. 18, no. 8, pp. 4157-4170, Aug. 2019.

\bibitem {Xie-21} H. Xie, J. Xu and Y. -F. Liu, ``Max-min fairness in IRS-aided multi-cell MISO systems with joint transmit and reflective beamforming," \emph{ IEEE Trans. on Wireless Commun.}, vol. 20, no. 2, pp. 1379-1393, Feb. 2021.

\bibitem {Pan-20}  C. Pan et al., ``Intelligent reflecting surface aided MIMO broadcasting for simultaneous wireless information and power transfer," \emph{ IEEE J. on Sel. Areas in Commun.}, vol. 38, no. 8, pp. 1719-1734, Aug. 2020.

\bibitem {Xu-20} D. Xu, X. Yu, Y. Sun, D. W. K. Ng and R. Schober, ``Resource allocation for IRS-assisted full-duplex cognitive radio systems," \emph{IEEE Trans. on Commun.}, vol. 68, no. 12, pp. 7376-7394, Dec. 2020.

\bibitem {Yu-19} X. Yu, D. Xu, and R. Schober, ``Enabling secure wireless communications via intelligent reflecting surfaces," \emph{in Proc. IEEE Global Commun. Conf. (GLOBECOM)}, Waikoloa, HI, USA, Dec. 2019, pp. 1-6.


\bibitem {Cui-19} M. Cui, G. Zhang, and R. Zhang, ``Secure wireless communication via intelligent reflecting surface," \emph{IEEE Wireless Commun. Letters}, vol. 8, no. 5, pp. 1410-1414, Oct. 2019.

\bibitem {Dong-22} L. Dong, H. -M. Wang and J. Bai, ``Active reconfigurable intelligent surface aided secure transmission," \emph{IEEE Trans. on Veh. Tech.}, vol. 71, no. 2, pp. 2181-2186, Feb. 2022.


\bibitem{Zheng-20} C. Zheng, W. Hao, P. Xiao, and J. Shi, ``Intelligent reflecting surface aided multi-antenna secure transmission," \emph{IEEE Wireless Commun. Letters}, vol. 9, no. 1, pp. 108-112, Jan. 2020.

\bibitem  {Chen-19} J.  Chen ,Y.  Liang , Y. Pei, and H. Guo, ``Intelligent reflecting surface: A programmable wireless environment for physical layer security," \emph{IEEE Access}, vol. 7, pp. 82599-82612, Jun. 2019.

\bibitem {Xu-19} D. Xu \emph{et al.}, ``Resource allocation for secure IRS-assisted multiuser MISO systems,"  \emph{2019 IEEE Globecom Workshops (GC WKshps)}, Waikoloa, HI, USA, Dec. 2019.

\bibitem {Dong-20b} L. Dong and H.-M. Wang, ``Enhancing Secure MIMO transmission via intelligent reflecting surface," \emph{IEEE Trans. Wireless Commun.}, vol. 19, no. 11, pp. 7543-7556, Nov. 2020.

\bibitem {Dong-20a} L. Dong, H.-M. Wang, ``Secure MIMO transmission via intelligent reflecting surface," \emph{IEEE Wireless Commun. Letters}, vol. 9, no. 6, pp. 787-790, Jun. 2020.

\bibitem {Dong-21} L. Dong, H. -M. Wang and H. Xiao, ``Secure cognitive radio communication via intelligent reflecting surface," \emph{IEEE Trans. Commun.}, vol. 69, no. 1 7, pp. 4678-4690, July 2021.

 \bibitem {Hong-20} S. Hong, C. Pan, H. Ren, K. Wang, and A. Nallanathan, ``Artificial-noise aided secure MIMO wireless communications via intelligent reflecting surface," \emph{ IEEE Trans. Commun.}, vol. 68, no. 12, pp. 7851-7866, Dec. 2020.

  \bibitem {Chu-21}  Z. Chu et al., ``Secrecy rate optimization for intelligent reflecting surface assisted MIMO system," \emph{IEEE Trans. Inf. Forensics Sec.}, vol. 16, pp. 1655-1669, Nov. 2021.

 \bibitem {Hong-21}  S. Hong, C. Pan, H. Ren, K. Wang, K. K. Chai and A. Nallanathan, ``Robust transmission design for intelligent reflecting surface-aided secure communication systems with imperfect cascaded CSI," \emph{IEEE Trans. on Wireless Commun.}, vol. 20, no. 4, pp. 2487-2501, April 2021.

 \bibitem {Xu-21} S. Xu, J. Liu and Y. Cao, ``Intelligent reflecting surface empowered physical layer security: signal cancellation or jamming?," \emph{IEEE Internet of Things Journal}, to appeared.


 \bibitem {Lu-20} X. Lu, W. Yang, X. Guan, Q. Wu and Y. Cai, ``Robust and secure beamforming for intelligent reflecting surface aided mmWave MISO systems," \emph{IEEE Wireless Commun. Lett.}, vol. 9, no. 12, pp. 2068-2072, Dec. 2020.

 \bibitem {Yu-20}  X. Yu, D. Xu, Y. Sun, D. W. K. Ng and R. Schober, ``Robust and secure wireless communications via intelligent reflecting surfaces," \emph{IEEE J. on Sel. Areas in Commun.}, vol. 38, no. 11, pp. 2637-2652, Nov. 2020.

\bibitem {Fang-21}  S. Fang, G. Chen and Y. Li, "Joint optimization for secure intelligent reflecting surface assisted UAV networks," \emph{IEEE Wirel. Commun. Lett.}, vol. 10, no. 2, pp. 276-280, Feb. 2021.

 \bibitem {Wang-19} H. -M. Wang, J. Bai and L. Dong, ``Intelligent reflecting surfaces assisted secure transmission without eavesdropper's CSI," \emph{IEEE Signal Process. Lett.}, vol. 27, pp. 1300-1304, 2020.


\bibitem {Hu-20} X. Hu, C. Zhong, Y. Zhang, X. Chen and Z. Zhang, ``Location information aided multiple intelligent reflecting surface systems," \emph{IEEE Trans. on Commun.}, vol. 68, no. 12, pp. 7948-7962, Dec. 2020.

 \bibitem {Zeng-20} S. Zeng, H. Zhang, B. Di, Z. Han and L. Song, ``Reconfigurable intelligent surface (RIS) assisted wireless coverage extension: RIS orientation and location optimization," \emph{ IEEE Commun. Letters}, vol. 25, no. 1, pp. 269-273, Jan. 2021.

 \bibitem {Mu-21} X. Mu, Y. Liu, L. Guo, J. Lin and R. Schober, ``Joint deployment and multiple access design for intelligent reflecting surface assisted networks," \emph{IEEE Trans. on Wireless Commun.}, vol. 20, no. 10, pp. 6648-6664, Oct. 2021.


 \bibitem {Wang-20} Z. Wang, L. Liu and S. Cui, ``Channel estimation for intelligent reflecting surface assisted multiuser communications: framework, algorithms, and analysis," \emph{ IEEE Trans. on Wireless Commun.}, vol. 19, no. 10, pp. 6607-6620, Oct. 2020.


 \bibitem {Wu-19c} Q. Wu and R. Zhang, ``Joint active and passive beamforming optimization for intelligent reflecting surface assisted SWIPT under QoS constraints," \emph{ IEEE J. on Sel. Areas in Commun.}, vol. 38, no. 8, pp. 1735-1748, Aug. 2020.

 \bibitem {Zeng-21}  S. Zeng, H. Zhang, B. Di, Z. Han and L. Song, ``Reconfigurable intelligent surface (RIS) assisted wireless coverage extension: RIS orientation and location optimization," \emph{IEEE Commun. Letters}, vol. 25, no. 1, pp. 269-273, Jan. 2021.

 \bibitem {Yan-16} S. Yan and R. Malaney, ``Location-based beamforming for enhancing secrecy in Rician wiretap channels," \emph{ IEEE Trans. on Wireless Commun.}, vol. 15, no. 4, pp. 2780-2791, April 2016.

\bibitem{Ma-14} S. Ma, M. Hong, E. Song, X. Wang and D. Sun,`` Outage constrained robust secure transmission for MISO wiretap channels," \emph{IEEE Trans. on Wireless Commun.}, vol. 13, no. 10, pp. 5558-5570, Oct. 2014.

\bibitem{Boyd-04} S. Boyd and L. Vandenberghe, \emph{Convex Optimization}, Cambridge University Press, 2004.

\bibitem{Cao-17}  P. Cao, J. Thompson and H. V. Poor, ``A sequential constraint relaxation algorithm for rank-one constrained problems," \emph{2017 25th European Signal Process. Conf. (EUSIPCO)}, 2017, pp. 1060-1064.

    \bibitem{T-16}  T. Lipp and S. Boyd, ``Variations and extension of the convex-concave procedure," \emph{Optim. Eng.}, vol. 17, no. 2, pp. 263-287, 2016.

\bibitem{Zhou-20} G. Zhou, C. Pan, H. Ren, K. Wang, and A. Nallanathan, ``A framework of robust transmission design for IRS-aided MISO communications with imperfect cascaded channels," \emph{IEEE Trans. Signal Process.}, vol. 68, pp. 5092-5106, 2020.



\end{thebibliography}
\end{document}